\documentclass[a4paper,11pt]{article}
\usepackage{oldgerm}
\usepackage{color}
\usepackage{bbm}
\usepackage{amssymb}
\usepackage{epsfig}
\usepackage{amsthm}
\usepackage{amsmath, mathrsfs,psfrag}

\renewcommand{\mathbb}[1]{\mathbbm{#1}}

\newcommand{\Cl}{\mathbbm{C}}
\newcommand{\Rl}{\mathbb{R}}
\newcommand{\Nl}{\mathbb{N}}

\definecolor{lightgray}{rgb}{0.8,0.8,0.8}

\newcommand{\Om}{\Omega}

\newcommand{\te}{\theta}

\newcommand{\la}{\lambda}

\newcommand{\eps}{\varepsilon}

\newcommand{\B}{\mathcal{B}}

\newcommand{\K}{\mathcal{K}}

\newcommand{\W}{\mathcal{W}}

\newcommand{\Oold}{\mathcal{O}}

\newcommand{\Ss}{\mathscr{S}}   

\newcommand{\frA}{\frak A}
\newcommand{\frF}{\frak F}

\newcommand{\bSigma}{\boldsymbol{\Sigma}}

\newcommand{\CCR}{\text{CCR}\,}
\newcommand{\CAR}{\text{CAR}\,}

\DeclareMathOperator{\supp}{supp}

\newcommand{\Sol}{{\rm Sol}}

\usepackage{amsmath}
\usepackage{amsfonts}
\usepackage{amssymb}
\usepackage{mathrsfs}
\usepackage[pdfborder={0 0 0}]{hyperref}
\usepackage{color}
\usepackage[T1]{fontenc}

\newtheorem{theorem}{Theorem}[section]
\newtheorem{proposition}[theorem]{Proposition}
\newtheorem{lemma}[theorem]{Lemma}
\newtheorem{corollary}[theorem]{Corollary}
\newtheorem{definition}[theorem]{Definition}

\newtheorem{example}[theorem]{Example}

\numberwithin{equation}{section}

\newlength{\dinwidth}
\newlength{\dinmargin}

\usepackage{mathtools, wrapfig,array}

\title{Linear hyperbolic PDEs with non-commutative time}
\author{
	 Gandalf Lechner and Rainer Verch\\
	\vspace*{-0mm}\\
	{\small Institut f\"ur Theoretische Physik, Universit\"at Leipzig, 04103 Leipzig, Germany}\\
	{\small\tt gandalf.lechner@uni-leipzig.de,  rainer.verch@itp.uni-leipzig.de}
}
\date{December 10, 2013}

\newcommand{\Cin}{{\mathscr{C}^\infty}}
\newcommand{\Ccin}{{\mathscr{C}^\infty_0}{}}
\newcommand{\Ltwon}{{\mathscr{L}^2}}

\begin{document}
\maketitle

\begin{abstract}
	Motivated by wave or Dirac equations on noncommutative deformations of Minkowski space, linear integro-differential equations of the form $(D+\la W)f=0$ are studied, where $D$ is a normal or prenormal hyperbolic differential operator on $\Rl^n$, $\la\in\Cl$ is a coupling constant, and $W$ is a regular integral operator with compactly supported kernel. In particular, $W$ can be non-local in time, so that a Hamiltonian formulation is not possible. It is shown that for sufficiently small $|\la|$, the hyperbolic character of $D$ is essentially preserved. Unique advanced/retarded fundamental solutions are constructed by means of a convergent expansion in $\la$, and the solution spaces are analyzed. It is shown that the acausal behavior of the solutions is well-controlled, but the Cauchy problem is ill-posed in general. Nonetheless, a scattering operator can be calculated which describes the effect of $W$ on the space of solutions of $D$. 
	
	It is also described how these structures occur in the context of noncommutative Minkowski space, and how the results obtained here can be used for the analysis of classical and quantum field theories on such spaces.
\end{abstract}

\section{Introduction}

Hyperbolic partial differential equations play a prominent role in many areas of physics, particularly in quantum field theory. They provide the dynamics for linear quantum field models which can be viewed as starting points, or building blocks, of any quantum field theory --- most importantly, of quantum field models describing interactions of
elementary particle physics (see e.g.\ \cite{Haag:1996,StreaterWightman:1964,BogolubovLogunovOksakTodorov:1990} for a synopsis). 

Moreover, hyperbolic partial differential equations are of similarly prominent importance when studying quantum fields  propagating on curved spacetime manifolds in semiclassical approaches to understanding the interplay of elementary particle physics and gravity \cite{Wald:1994,HollandsWald:2010,BeniniDappiaggiHack:2013}. The mathematical questions related to such differential equations, such as theorems on the global well-posedness of the Cauchy problem and on the causal propagation character of solutions of hyperbolic partial differential operators, and Dirac operators, on globally hyperbolic Lorentzian manifolds, are well understood (see, for example \cite{BarGinouxPfaffle:2007}).

However, the situation is quite different when considering non-commutative (or non-local) modifications of hyperbolic equations, for example
\begin{align*}
	\Box f+w\star f=0,
\end{align*}
where $\Box$ is the d'Alembertian and $\star$ a non-commutative product between $w$ and $f$, typically given by an integral expression. Similar differential equations have been studied before, even in the non-linear case \cite{DurhuusGayral:2010}. However, the novel point here is that we consider the case where the product $\star$ also involves integration in the time coordinate (``non-commutative time''). Then a Hamiltonian formulation is not possible, and the usual theorems referred to above do not apply. Nonetheless, such equations appear in the context of field theories on noncommutative spacetimes; and their analysis is the topic of the present paper. 

To explain our setting and motivation in more detail, it is instructive to first recall some well-known facts on the solutions and dynamics of linear hyperbolic differential equations without non-local perturbations. 

One of the easiest ways of introducing an interaction in a field theory --- be it a classical field or a quantum field --- is to couple the field to an external potential. Given sufficient regularity conditions that the external potential has to obey, one finds under very general circumstances that there is a scattering operator which compares the dynamics
of solutions to a hyperbolic wave equation with an external potential to those without. To explain this more specifically, assume that we are given some second order linear hyperbolic differential operator
\begin{align*}
	D:C^\infty(\Rl^n,\Cl^N)\to C^\infty(\Rl^n,\Cl^N)
\end{align*}
whose principal part is the d'Alembertian $\Box = \partial^2/\partial x_0^2- \sum_{k=1}^{n-1} \partial^2/\partial x_k^2$.
\footnote{The operator is implicitly regarded as a second order partial differential operator on Minkowski spacetime with metric principal part, hence we use the convention to denote elements of $\Rl^n$ as $x = (x_0,x_1,\ldots,x_{n-1})$
where $x_0$ is viewed as time-coordinate and the $x_k$, $k=1,\ldots,n-1$ are spatial coordinates.} We regard $D$ as the ``unperturbed'' hyperbolic differential operator. It defines a linear wave equation
\begin{align*}
	Df=0
	\,,
\end{align*}
and any $f\in C^\infty(\Rl^n,\Cl^N)$ fulfilling this equation and having the property that its restrictions to hyperplanes of constant time $x_0$ have compact support will hence be called a solution of the unperturbed wave equation with compactly supported Cauchy data.
We collect all such solutions in a set $\Sol_0$ (which naturally is a vector space). 
Next, we can ``perturb'' the hyperbolic differential operator $D$ by adding some other linear map $W :C^\infty(\Rl^n,\Cl^N) \to C^\infty(\Rl^n,\Cl^N)$ to obtain perturbed operators
\begin{align*}
	D_\lambda = D + \lambda W
	\,,
\end{align*}
parametrized by $\lambda \in \Cl$. Thus, $W$ is viewed as an external potential or ``perturbation'', and $\lambda$ is a parameter which (if taken real and positive) scales the coupling strength of the fields. At this stage, $W$ can be very general, but
let us first restrict to the case that $W$ acts as a pointwise (matrix-valued) multiplication operator, i.e. $(Wf)(x) = w(x)f(x)$ with $w\in C^\infty(\Rl^n,\Cl^{N\times N})$. In this case, $D_\lambda$ is again a second order linear hyperbolic partial differential operator; we call $D_\lambda f = 0$ the perturbed wave equation (if $\lambda \ne 0$) and we denote by $\Sol_\lambda$ the set of smooth solutions to this differential equation which have compactly supported Cauchy data, in complete analogy to the case of $D$ before.

Under these assumptions, the Cauchy problem for the wave equation $D_{\lambda}f = 0$ is well-posed, i.e. for any $C_0^\infty$-Cauchy data, it has a unique solution. Moreover, there are uniquely determined advanced and retarded Green's operators $R_\lambda^\pm: C_0^\infty(\Rl^n,\Cl^N) \to C^\infty(\Rl^n,\Cl^N)$ such that $D_\lambda R_\lambda^\pm g = g = R_\lambda^\pm D_\lambda g$ for all $ g \in C_0^\infty(\Rl^n,\Cl^N)$ and $\supp(R_\lambda^\pm g) \subset J^{\pm}(\supp g)$, where $J^\pm({\rm supp}\,g)$ is the causal future(+), resp.\ causal past$(-)$ of $\supp g$. Defining the causal propagator $R_\lambda = R^-_\lambda - R_\lambda^+$, it then holds that  $R_{\lambda}$ maps $C_0^\infty(\Rl^n,\Cl^N)$ onto $\Sol_\lambda$, see Sec.~\ref{Section:Basics} and the literature cited there for details. As an aside, we remark that this holds more generally for the case that $D_\lambda$ is a second order 
linear operator with metric principal part on a globally hyperbolic Lorentzian manifold, or that $D_\lambda$ is a Dirac-type operator, see \cite{Dimock:1982}. Therefore,
our restriction here to the case of $n$-dimensional Minkowski spacetime is mainly for the sake of simplicity.

With these very general statements about the solutions to the wave equation $D_\lambda f = 0$, one is also in the position to define a scattering operator alluded to above.  To this end, one first introduces the M\o ller operators. The M\o ller operators and correspondingly the scattering operator are known to exist under quite general conditions on the potential $w$, a sufficient (but not necessary) condition is $w\in C_0^\infty(\Rl^n,\Cl^{N\times N})$. In this situation, there are $\tau_\pm \in \Rl$ so that $w$ vanishes on the two future/past regions $\Sigma_{\tau_{\pm}}^\pm:=\{ x \in \Rl^n: \pm x_0 > \pm \tau_\pm \}$. The M\o ller operators are then defined by 
\begin{align*}
		\Omega_{\lambda,\pm}: \Sol_\lambda \to \Sol_0\,, \quad R_\lambda g  \mapsto R_0 g
		\,,\qquad
		\supp g\subset\Sigma^\pm_{\tau_\pm}\,,
\end{align*}
and the scattering operator by
\begin{align*}
	S_\lambda = \Omega_{\lambda,+} (\Omega_{\lambda,-})^{-1} : \Sol_0 \to \Sol_0 \,.
\end{align*}
The action of the scattering operator compares solutions of the perturbed wave equation to solutions of the unperturbed wave equation in the following way. First, a given solution of the unperturbed wave equation is propagated ``backwards in time'' to the spacetime region $\Sigma_{\tau_-}^-$ where the perturbation $W$ vanishes. Therefore, in this spacetime region, the solution of the unperturbed wave equation is also a solution to the perturbed wave equation, and thus (by the well-posedness of the Cauchy problem) it determines a solution to the perturbed wave equation
on all of Minkowski spacetime. That solution to the perturbed wave equation is propagated ``forward in time'' to the spacetime region $\Sigma_{\tau^+}^+$ where again, $W$ vanishes. The solution to the perturbed wave equation therefore determines a new solution to the unperturbed wave equation on all of Minkowski spacetime. The scattering operator assigns this new solution to the unperturbed wave equation to the one initially given. Equivalently, $S_\la$ can be viewed as mapping the past asymptotics of a solution of $D_\la$ to its future asymptotics, and thus provides partial information about the dynamics described by $D_\la$. We note that this way of describing the scattering process is analogous to the concept of ``relative Cauchy evolution'' which was 
studied in the context of local general covariant quantum field theory \cite{BrunettiFredenhagenVerch:2001,FewsterVerch:2012,HollandsWald:2001}.

Upon choosing a Cauchy data formulation of the wave equation, the scattering operator can, in fact, be cast in a perhaps more familiar form, similar to the way a scattering operator appears in quantum mechanics. Let $u \in C_0^\infty(\Rl^{n-1},\Cl^N) \times   C_0^\infty(\Rl^{n-1},\Cl^N)$ denote some Cauchy data for any of the wave equations at $x_0 = 0$. Let $f_\lambda[u]$ be the solution of the
wave equation, $D_\lambda f_\lambda[u] = 0$ having these Cauchy data at $x_0 = 0$. Then we denote by 
\begin{align*}
	T_{\la,t}:C_0^\infty(\Rl^{n-1},\Cl^N) \times   C_0^\infty(\Rl^{n-1},\Cl^N)
	\to C_0^\infty(\Rl^{n-1},\Cl^N) \times   C_0^\infty(\Rl^{n-1},\Cl^N)
\end{align*}
the operator which assigns to $u$ the Cauchy data of $f_\lambda[u]$ at $x_0 = t$. In this setting, the M\o ller operators are
\[ \Omega_{\lambda,\pm} = \lim_{t \to \pm \infty} T_{0,t}(T_{\lambda,t})^{-1} \]
(in fact, the limits are already assumed for finite $t$ such that $\pm t > \pm \tau_{\pm}$), and again the scattering operator is $S_\lambda =  \Omega_{\lambda,+} (\Omega_{\lambda,-})^{-1}$.
  
One can therefore see that the class of perturbations $W$ to the linear hyperbolic wave operators can be considerably enlarged when one adopts the point of view to put the wave equation into its Cauchy data form (or, as it is sometimes called, in Hamiltonian form): Then the wave equation is re-formulated as a one-parametric evolution equation
\begin{align} \label{eq:evoleq}
	\frac{d}{dt}u_t + A_t u_t & = 0\,,
\end{align}
where now $u_t = T_{0,t}u$ is in the Cauchy data space $C_0^\infty(\Rl^{n-1},\Cl^N) \times C_0^\infty(\Rl^{n-1},\Cl^N)$ and for any real $t$, $A_t$ is a linear operator on that space. Clearly, this kind of evolution equation can immediately be generalized to the case where the Cauchy data space is replaced by a suitable Sobolev space or Hilbert space, or a more general type of topological vector space. If \eqref{eq:evoleq} is the evolution equation corresponding to a wave equation, then the $A_t$ are partial differential operators and $A_t$ is local in the sense that $\supp(A_t u) \subset \supp(u)$, but that is in fact not necessary in order to obtain statements on the well-posedness of the initial value problem for such evolution equations. This circumstance makes the class of problems which can be treated in form of an evolution equation of the form \eqref{eq:evoleq} very large, as is well known (see, e.g. \cite{SellYou:2002} and literature cited therein as just one sample reference).

In particular, assuming that \eqref{eq:evoleq} corresponds to the evolution equation of the wave equation $Df = 0$, then one can treat perturbation operators $W$ whose corresponding evolution equations assume the form
\begin{align} \label{eq:evoleqpert}
 \frac{d}{dt}u_{\lambda,t} + (A_t + \lambda B_t) u_{\lambda,t} & = 0 \,,
\end{align}
where the perturbation operators $B_t$ can be fairly arbitrary, apart from the requirement that one would indeed still like to assert the well-posedness of the initial value problem for this equation (if possible globally, i.e.\ with solutions being defined, and being sufficiently regular, for all $t \in \Rl$). This means that the $B_t$ need not be local operators, but can be quite general pseudodifferential operators, or integro-differential operators,
for example.

However, what is clearly required in this approach is that $W$ acts ``locally in time'' so as to allow an equivalent rewriting of the wave equation $D_\lambda f = 0$ in terms of \eqref{eq:evoleqpert} where the perturbation $W$ can be re-expressed by a family of operators $B_t$, each acting on the Cauchy data $u_{\lambda,t}$ of the solution $f$ at time $x_0 = t$. There exist perturbations where $W$ is not local in time in this sense and where, hence, the perturbed wave equation $D_\lambda f = 0$ cannot be cast into the form of an evolution equation of the type \eqref{eq:evoleqpert}. Typical examples are integral operators involving integration over the time variable.

It is our main aim to analyze such non-local in time hyperbolic differential equations. While we are restricting ourselves to the mathematical core of the problem in this article, it might be useful to summarize here together with our results also our motivations, which are related to quantum physics, and more particularly quantum field theory on noncommutative spacetimes. In fact, examples of non-local in time perturbations are given by certain noncommutative multiplication operators, such as the Moyal product \cite{GayralGraciaBondiaIochumSchuckerVarilly:2004}. Let us suppose that $n$ is even, and that $\theta$ is some fixed, antisymmetric real invertible $(n \times n)$ matrix. Then let
$(\alpha_z f)(x) = f(x+z)$ be the action of the translations on measurable functions $g$ on $\Rl^n$. The Moyal-Rieffel
product between smooth functions is given by
\begin{align}\label{eq:RieffelProduct}
	w \star f = \int_{\Rl^{n}}dp\, \int_{\Rl^n}dz\, e^{2\pi i(p,z)} \alpha_{\theta p}w \cdot \alpha_z f
	\,,
\end{align}
understood as an oscillatory integral on, say, Schwartz space. This product between functions is (owing to $\theta$ having full rank) non-local in time (and also in space) even in the sense that, if $w$ and $f$ both have compact support, then in general $w \star f$ will not have compact support in any of its $x_\mu$-variables. 

The Moyal-Rieffel product and variants thereof play an important role in deformation quantization \cite{Waldmann:2007}, and also in the context noncommutative Minkowski spacetime, see for example the review \cite{Szabo:2003}. However, before presenting our vantage point as to why we think that wave equations with non-local in time perturbations of the type provided by the Moyal-Rieffel product are useful when studying quantum fields on non-commutative spacetimes, it may be more instructive to first present the assumptions made and results achieved in the present work. In describing this, we will also summarize the content of this article.

Our starting point is a wave operator $D$ as before, or a first oder linear partial differential operator $D$ for which there is another first order linear partial differential operator $D'$ such that $D'D$ is a wave operator. Precise definitions appear in Sec.~\ref{Section:Basics}. Known results on the existence and uniqueness of advanced and retarded Green's operators $R^\pm$ for $D$ will also be summarized there. Then we consider a perturbation $W$ of $D$ which is a $C_0^\infty$-kernel operator, i.e.
\[ (Wf)(x) = \int dy\, w(x,y) f(y) \]
with $w \in C_0^\infty(\Rl^n \times \Rl^n, \Cl^{N\times N})$, and as before, we consider the perturbed operators
\[ D_\lambda = D + \lambda W \,.\]
The main result of Sec.~\ref{Section:GreenAndCauchy} is that, provided that $|\lambda|$ is sufficiently small, there exist fundamental solutions (Green's operators) $R_\lambda^\pm$ for $D_\lambda$, distinguished by the property that 
\[ D_\lambda R_\lambda^\pm f = f = R_\lambda^\pm D_\lambda f \]
for all $f \in C_0^\infty(\Rl^n,\Cl^N)$, and a localization property of $\supp(R_\la^\pm f)$ which is dominated by the causal propagation of the unperturbed fundamental solutions $R^\pm$. Moreover, we show that such Green's operators for $D_\lambda$ are unique. For Cauchy data imposed outside of the support region of $w$, i.e. outside the causal closure of some set $K$ such that $\supp w\subset K\times K$, the Cauchy problem for $D_\lambda f = 0$ turns out to be well-posed. In contrast, we will also show that there are $C_0^\infty$-kernels $w$ such that the Cauchy problem for $D_\lambda f=0$ is ill-posed; both the existence and uniqueness of its solutions can break down in this situation.

Nonetheless, we will establish that, provided $|\lambda|$ is sufficiently small, it is possible to define a scattering
operator, relying on the well-posedness of $D_\lambda f = 0$ for Cauchy data imposed outside of the support region of $w$.  

These findings allow, from our point of view, to draw the following conclusion: If the perturbation $W$ is non-local in time, one can in general not expect that the resulting wave equation $D_\lambda f = 0$ permits a well-posed Cauchy problem, or that its solutions propagate strictly causally --- in particular, not for arbitrary $\lambda$. Only for small coupling, $\lambda W$ can be considered as
a small perturbation of the wave operator $D$ so that the dynamics is still mainly determined by the hyperbolic character
of $D$ and therefore admits unique advanced and retarded Green's operators $R_\lambda^\pm$. When $|\lambda|$ is made
large, it may happen that the hyperbolic character of $D$ is no longer the dominating contribution to the dynamics. Therefore, we think that the scattering of Cauchy data in the past of $K$ to the future of $K$, described by the scattering operator $S_\lambda$, which exists for sufficiently small $|\lambda|$, should actually be seen as the generalization, or replacement, of the well-posedness of the Cauchy problem for the wave equation $(D + \lambda W)f = 0$ when $W$ is non-local in time.

We also show that the ``generator'' of the scattering operator $S_\lambda$ with respect to variation of the coupling parameter
$\lambda$ is given by $RW$, i.e.
\[ \left. \frac{d(S_\lambda f_0)}{d \lambda} \right|_{\lambda = 0} = RW f_0 \]
for all $f_0 \in {\rm Sol}_0$. While this quantity might seem a quite weak replacement for the full dynamics as described by the Cauchy problem, it is actually an important object in the context of field theory, as recalled below.

In Sec.~\ref{Section:Quantization}, we consider the quantization of (solutions of) the wave equations, or Dirac-type equations, in terms of assigning abstract CCR-algebras or CAR-algebras to the corresponding solution spaces. This procedure is entirely standard;
we establish that, under suitable --- very general --- conditions, the scattering operators $S_\lambda$ defined on
${\rm Sol}_0$ for underlying $C_0^\infty$-kernel operators $W$ and for sufficiently small $|\lambda|$ induce $C^*$-algebra morphisms
$s_\lambda$ on the associated CCR or CAR algebras of the quantized fields. They are also called ``scattering morphisms'' or
``scattering Bogoliubov transformations''. 

In Sec.~\ref{Section:StarProducts}, we consider perturbation terms given by certain star products, i.e. $Wf=w\star f$ with a product $\star$ of the form \eqref{eq:RieffelProduct}. Such perturbations are limits of $C_0^\infty$-kernel operators, and we discuss two examples: The Moyal product itself, corresponding to the canonical translation action $\alpha$, and a ``locally noncommutative star product'' as discussed in \cite{HellerNeumaierWaldmann:2006,BahnsWaldmann:2007,LechnerWaldmann:2011}, in which the action $\alpha$ only acts non-trivially in a bounded region. Both these perturbations do not have smooth and compactly supported integral kernels (in each example, one of the two properties fails), but one always has a family $W_\eps$ of $C_0^\infty$-kernel operators such that $W_\eps\to W$ for $\eps\to0$. One thus obtains, for each positive $\eps$, scattering operators $S_{\eps,\lambda}$ provided $|\lambda|$ is sufficiently small (possibly depending on $\eps$). Nevertheless, the ``generators'' $d(S_{\eps,\lambda})/d\lambda |_{\lambda = 0} = RW_\eps$ are independent of $\lambda$, and one can argue that in the context of quantum field theory, these are mainly the objects one is interested in. For the case of Dirac operators and the Moyal product, such an investigation can also be found in the thesis \cite{Borris:2011}.

This brings us back to the motivation for considering perturbation operators $W$ arising from Moyal-Rieffel-type products
in the context of quantum fields on non-commutative spacetimes; here we refer to discussion in \cite{BorrisVerch:2010} and \cite{Verch:2011} where such ideas are discussed in more detail.  Suppose that we have a quantum field on Minkowski
spacetime, i.e. an operator-valued distribution $\psi(f)$, $f \in C_0^\infty(\Rl^n,\Cl^N)$, fulfilling the wave equation
$\psi(Df) = 0$; concretely, in line with \cite{BorrisVerch:2010,Verch:2011}, we think of the Dirac field, so that it fulfills the CAR relations (cf. Sec.~\ref{Section:Quantization})
\[ \psi(f_1)^*\psi(f_2) + \psi(f_2)\psi(f_1)^* = i \langle f_1,\gamma^0 R f_2 \rangle \cdot 1\,. \]
A natural way of setting up a quantum field like the Dirac field on a ``non-commutative'' spacetime, such as the Moyal-Rieffel deformed Minkowski-spacetime, is to replace the standard product of field operators by the corresponding Rieffel-deformed product with respect to some action of the translation group. Either on a classical or quantum level, such a procedure is contained in most approaches to quantum field theory on noncommutative spacetime, see for example the articles \cite{DoplicherFredenhagenRoberts:1995,DouglasNekrasov:2001,Szabo:2003,GrosseWulkenhaar:2005,Rivasseau:2007,Soloviev:2008,GrosseLechner:2008} and the literature cited therein. On the quantum level, this is equivalent to ``deforming'' the quantum field operators by the procedure of ``warped convolution'' \cite{GrosseLechner:2007,BuchholzSummers:2008,GrosseLechner:2008,BuchholzLechnerSummers:2011}. 

There are however many different approaches to the task of defining quantum field theories on noncommutative spacetimes, one other notable possibility being to regard Moyal-Rieffel deformed Minkowski
spacetime in the light of a ``Lorentzian spectral triple'', in the spirit of Connes' spectral geometry \cite{Connes:1990,GraciaBondiaVrillyFigueroa:2001}. Also here, one of the essentials is to replace the commutative algebra of scalar ``test-functions'' on Minkowski spacetime $\Rl^n$, endowed with the usual pointwise product $h\cdot g$, by the non-commutative algebra of ``test-functions'', where the pointwise product is replaced by the Moyal-Rieffel product $h \star g$, or a variant thereof. One may wonder how that changing of products takes effect on the quantum field operators, and one possible response to that question is: by different forms of scattering of the quantized Dirac field by an external, scalar potential.

On the usual (``undeformed'') Minkowski spacetime, the coupling of the field to an external potential
is by the commutative pointwise product of functions, e.g.\ the corresponding perturbation operator for the Dirac equation is 
$Wf=w\cdot f$, $f \in C_0^\infty(\Rl^n,\Cl)$, for some function $w \in C_0^\infty(\Rl^n,\Rl)$, say. Then one gets the scattering operator $S_\la$ for the corresponding $D_\la$ and a scattering morphism $s_\la$ on the algebra of Dirac field operators $\psi(f)$, which results in \cite{BorrisVerch:2010}
\begin{align*}
	\left. \frac{d}{d\lambda} s_\lambda(\psi(f))\right|_{\lambda=0} =
	\psi(w\cdot Rf) = i[:\psi^+\psi:(w),\psi(f)]
	\,,
\end{align*}
where on the right hand side, $:\psi^+\psi:(w)$ is the Wick-product (in vacuum representation) of the adjoint  quantized Dirac field and the quantized Dirac field --- which is a scalar quantum field, hence observable, with the interpretation of a quantized field density.

If we consider an external scalar potential on the Moyal-Rieffel deformed Minkowski spacetime, then the perturbation
operator $Wf = w \star f$ is the limit as $\eps\to0$ of a sequence of $C_0^\infty$-kernel operators $W_\eps$. In the case of a Moyal product with commutative time (where a Hamiltonian formulation of the differential equation is at hand), it was shown in \cite{BorrisVerch:2010} that the scattering morphism $s_\lambda$ takes the form
 \[ \left. \lim_{\eps \to 0}\frac{d}{d\lambda} s_{\eps,\lambda}(\psi(f))\right|_{\lambda = 0} = \psi(w \star Rf) = i[:\psi^+\psi:(w)\,\overset{\star}{,}\,\psi(f)] \]
where on the right hand side, there is a suitably Moyal-Rieffel deformed version of the commutator. More interesting, however, is the fact that one can expect that there are (symmetric) operators $X(w)$ such that
   \[  \left. \lim_{\eps \to 0}\frac{d}{d\lambda} s_{\eps,\lambda}(\psi(f))\right|_{\lambda = 0} = i[X(w),\psi(f)]\,, \]
and these operators would obviously be different from the operators $:\psi^+\psi:(w)$; in particular, they have completely different localization properties. In the spirit of ``Bogoliubov's formula'' \cite{BogolubovShirkov:1980}, the operators $X(w)$ should actually be regarded
as ``observables'' of the quantum field associated with $w$ regarded as an element of the Moyal-Rieffel deformed ``test-functions'' over Minkowski spacetime, see \cite{BorrisVerch:2010} for further discussion. Following this line of argument offers a systematic way to associate quantum field observables to elements of a ``algebra of non-commutative coordinates'' seen, at least tentatively, from
a perspective of Connes' spectral geometry. 

With the analysis of non-local in time perturbations carried out in the present paper, one can expect to obtain similar results also in the case of fully noncommutative spacetime, including noncommutative time. In fact, since the warped convolution defines a representation of the Rieffel-deformed product \cite{BuchholzLechnerSummers:2011}, the $X(w)$ should be related to warped convolutions of the $:\psi^+\psi(w):$. We therefore expect that our results will be helpful in comparing different approaches to quantum field theory on noncommutative spacetimes, and provide tools to extract the relevant physical effects. These questions will be investigated elsewhere.

\section{(Pre-)normally hyperbolic differential operators with smoothing compactly localized perturbations}\label{Section:Basics}

\subsection{Basic Definitions and Preliminaries}

In this subsection, we set up our notation and introduce the objects of our investigations. Starting with the basic geometric data, Minkowski space $\Rl^n$, $n\geq2$, is endowed with its standard metric of signature $+-\ldots -$, and we put $s:=n-1$. The causal future $(+)$/past $(-)$ of a set $B\subset\Rl^n$ is denoted $J^\pm(B)$. In particular, $V^\pm:=J^\pm(\{0\})$ is  the forward/backward light cone. A subset $M$ of Minkowski space is called {\em causally convex} if for any causal curve in $\Rl^n$ with endpoints in $M$, the whole curve lies in $M$. The term {\em Cauchy hyperplane} is used to refer to Cauchy surfaces of the form $\Sigma=\Sigma_t=\{t\}\times\Rl^s$, $t\in\Rl$, and we also employ the short hand notations $\Sigma_t^\pm$ for the interior of $J^\pm(\Sigma_t)$.

We write
\begin{align*}
	\Cin := C^\infty(\Rl^n,\Cl^N)
	\,,\quad
	\Ccin := C_0^\infty(\Rl^n,\Cl^N)
	\,,\quad
	\Ltwon := L^2(\Rl^n,\Cl^N)
	\,,
\end{align*}
for the space of smooth functions, smooth functions of compact support, and (equivalence classes of) Lebesgue square integrable functions $f:\Rl^n\to\Cl^N$, respectively, where $N\in\Nl$ is some integer. For subsets $K\subset\Rl^n$, we write $\Cin(K)$ for the subspace of $\Cin(K)$ of functions with support in $K$, and analogously for the other spaces. All these spaces are endowed with their standard topologies, see e.g. \cite{Treves:1967}. The scalar product in $\Ltwon$ is denoted $\langle\,\cdot\,,\,\cdot\,\rangle$, and the one in $\Cl^N$ by $(\,\cdot\,,\,\cdot\,)$, i.e. $\langle f,g\rangle =\int dx\,(f(x),g(x))$. The associated norms are $\|\cdot\|_2$ and $|\cdot|$, respectively. 

By ``differential operator'' we shall always mean linear finite order differential operator with smooth coefficients, and by ``Cauchy data $u$ on some Cauchy hyperplane $\Sigma$'', we shall always mean smooth compactly supported Cauchy data (sometimes also called $\Ccin$-Cauchy data). In the context of a first order (in time) differential operator, this is a function $u\in \Ccin(\Sigma)$, whereas in the context of a second order operator, this is a pair of functions, $u\in \Ccin(\Sigma)\times\Ccin(\Sigma)$.
\\\\
\indent In the following, we will study differential operators of the form $D_\la=D+\la\,W$, where $D$ is a (pre-)normally hyperbolic differential operator, $\la\in\Cl$ a coupling constant, and $W$ some non-local perturbation term. These objects are defined next.

\begin{definition}\label{definition:PreNormallyHyperbolic}{\bf ((Pre-)normally hyperbolic differential operators)}
	\begin{enumerate}
		\item A linear differential operator $D$ on $\Cin$ is called normally hyperbolic if there exist smooth matrix-valued functions $U^0,...,U^s,V:\Rl^n\to\Cl^{N\times N}$ such that
		\begin{align}
			D
			=
			\frac{\partial^2}{\partial x_0^2}
			-
			\sum_{k=1}^s\frac{\partial^2}{\partial x_k^2}
			+
			\sum_{\mu=0}^s U^\mu(x)\frac{\partial}{\partial x_\mu}
			+V(x)
			\,.
		\end{align}
		\item A linear differential operator $D$ on $\Cin$ is called prenormally hyperbolic if $D$ is of first order, and there exists another first order differential operator $D'$ on $\Cin$ such that $D'D$ is normally hyperbolic.
	\end{enumerate}
\end{definition}

Whereas the first definition is standard in the context of wave equations (see for example \cite{BarGinouxPfaffle:2007, Waldmann:2012}), the second one is taken from \cite[Def.~1]{Muhlhoff:2011}, and basically tailored towards a convenient description of Dirac operators. In fact, the Dirac operator $D=-i\sum_{\mu=0}^s\gamma^\mu\frac{\partial}{\partial x_\mu}+V(x)$ is prenormally hyperbolic when the matrices $\gamma^\mu$ generate an irreducible representation of the complexified Clifford algebra $\Cl l_{1,s}$ \cite{Thaller:1991,Coquereaux:1988}, and $V:\Rl^n\to\Cl^{N\times N}$ is smooth. We also mention that with $D$, also $D'$ is prenormally hyperbolic.

For our analysis, both the normally hyperbolic and the prenormally hyperbolic case are equally well suited, and we will thus consider an arbitrary (pre-)normally hyperbolic differential operator $D$.

In \cite{Muhlhoff:2011}, it was shown that prenormally hyperbolic operators inherit many of the well-known properties of normally hyperbolic ones \cite{BarGinouxPfaffle:2007}, see also \cite{Dimock:1982, Sanders:2009-2, BorrisVerch:2010} for corresponding arguments for Dirac operators. We summarize here the characteristic features of (pre-)normally hyperbolic operators we will rely on in this article.

\begin{theorem}\label{Theorem:PropertiesOfPreNormallyHyperbolicOperators}{\bf (Properties of (pre-)normally hyperbolic differential operators)}\\
	Let $D$ be a (pre-)normally hyperbolic differential operator. Then 
	\begin{enumerate}
		\item[D1)]\label{item:Hyp-LinCont} $D:\Cin\to\Cin$ is a linear continuous map.
		\item[D2)]\label{item:Hyp-Local} $D$ is local in the sense that $\supp(Df)\subset\supp f$ for any $f\in\Cin$.
		\item[D3)]\label{item:Hyp-Cauchy} For any Cauchy hyperplane $\Sigma$ and any Cauchy data $u$ on $\Sigma$, there exists unique $f_0[u]\in\Cin$ such that $Df_0[u]=0$ and $f_0[u]|_\Sigma=u$ (for prenormally hyperbolic $D$), respectively $f_0[u]|_\Sigma=u_1$, $\partial_0 f_0[u]|_\Sigma=u_2$, $u=(u_1,u_2)$ (for normally hyperbolic $D$).
		\item[D4)]\label{item:Hyp-Green} $D$ has unique advanced and retarded fundamental solutions, i.e. continuous linear maps $R^\pm:\Ccin\to\Cin$ uniquely determined by the conditions that $DR^\pm f=f=R^\pm D f$ and $\supp (R^\pm f)\subset J^\pm(\supp f)$ for all $f\in\Ccin$.
	\end{enumerate}

	\noindent We define the causal propagator as\footnote{Also called ``Green's operator''. Note that the sign convention used here differs from the one in the previous article \cite{BorrisVerch:2010}.}
	\begin{align}
		R:=R^--R^+
		\,,
	\end{align}
	and the space of all solutions with compact support on Cauchy hyperplanes as
	\begin{align}
		\Sol_0
		:=
		\{f\in\Cin\,:\,D f=0\,,\quad \supp f \subset (V^++a_+)\cup(V^-+a_-)\,\text{ for some } a_\pm\in\Rl^n\}
		\,.
	\end{align}

	\noindent We then furthermore have
	\begin{enumerate}
		\item[D5)]\label{item:Hyp-Adjoint} The formal adjoint of $D$ w.r.t. $\langle\,\cdot\,,\,\cdot\,\rangle$, denoted $D^*$, is also (pre-)normally hyperbolic and thus satisfies D1)--D4) as well. The retarded/advanced fundamental solutions of $D^*$, denoted $S^\pm$, are related to $R^\pm$ by $S^\pm=(R^\mp)^*$.
		\item[D6)]\label{item:Hyp-Sol} Let $\Sigma$ be a Cauchy hyperplane and $\bSigma$ an open causally convex neighborhood of $\Sigma$. Then 
		\begin{align}
			\Sol_0 = R\Ccin(\bSigma)\,.
		\end{align}
		\item[D7)]\label{item:Hyp-Restrict} For any causally convex subset $M\subset\Rl^n$, the restriction of $D$ to $M$ satisfies properties analogous to D1)--D6). The retarded and advanced fundamental solutions $R_M^\pm$ of $D|_M:\Cin(M)\to\Cin(M)$ are the restrictions (in domain and range) of the $R^\pm$, $R^\pm_M:\Ccin(M)\to\Cin(M)$.
	\end{enumerate}
\end{theorem}

Items {\em D1)} and {\em D2)} hold for all linear differential operators, and the proofs of {\em D3)--D7)} for the normally hyperbolic case can be found in \cite{BarGinouxPfaffle:2007}. For the prenormal case, {\em D3)--D4)} have been proven explicitly by M\"uhlhoff \cite{Muhlhoff:2011}, and also {\em D5)--D7)} can be quickly extracted from his construction of fundamental solutions.

The properties {\em D1)--D7)} are not all independent of each other. For example, {\em D7)} can be deduced from the uniqueness of the fundamental solutions, and leads to a functorial assignment from the category of all globally hyperbolic sub-spacetimes of $\Rl^n$, with isometric embeddings as arrows, to the corresponding solution spaces \cite{BrunettiFredenhagenVerch:2001,BarGinouxPfaffle:2007,FewsterVerch:2012}.
\\
\\
The second essential input in our analysis is a non-local perturbation term. As explained in the Introduction, we are interested in describing star product multipliers, or approximate version thereof, which suggests to consider integral operators with $\Ccin$-kernels. 

Let us now motivate this choice also from a mathematical perspective. Given a linear map $W:\Cin\to\Cin$, a smooth function $f\in\Cin$ can be a solution of the perturbed operator $D_\la=D+\la\,W$ only if $Wf\in D\Cin$. In the extreme case that $W\Cin\cap D\Cin=\{0\}$, $f$ is a solution of $D_\la$ if and only if $Df=0$ and $Wf=0$ separately. Such solutions can exist, and there are even examples where any solution of $Df=0$ automatically also satisfies $Wf=0$. However, these solutions are uninteresting from our point of view, as they are just solutions of the unperturbed equation and in particular do not depend on the coupling $\la$. In this situation, $D$ and $W$ completely decouple, and the scattering at $W$ will be trivial.

We will therefore rather consider situations where $W\Cin\subset D\Cin$, where an interesting solution theory for $D_\la$ is not ruled out from the beginning. As $\Ccin\subset D\Cin$ by the existence of Green's operators $R^\pm$ postulated above, this will in particular be the case when $W\Cin\subset\Ccin$. In this situation, any $f\in\Cin$ satisfies $D_\la f=Df+\la Wf=D(f+\la R^\pm Wf)$, which vanishes if $(1+\la R^+W)f=Rh$ for some $h\in\Ccin$. Hence many solutions will exist if $(1+\la R^\pm W)$ can be inverted. Formally, the inverse is given by $(1+\la R^\pm W)^{-1}=\sum_{k=0}^\infty (-\la R^\pm W)^k$, but convergence of this series in a useful topology is not automatic. 

In the next section, we will study this question in an $\Ltwon$-setting, and to this end, it is necessary that $W$ is regular enough to make $R^\pm W$ bounded in an $\Ltwon$-operator norm. These requirements can most easily be met when taking $W$ to be a $\Ccin$-kernel operator.

\begin{definition}\label{Definition:CInftyIntegralOperator}{\bf ($\Ccin$-kernel operator)}\\
	A $\Ccin$-kernel operator is a mapping $W:\Cin\to\Cin$ which can be represented as
	\begin{align}\label{eq:CinftyIntegralOperator}
		(Wf)(x)
		:=
		\int dy\,w(x,y)f(y)
		\,,\qquad f\in\Cin,
	\end{align}
	where $w\in\Ccin(\Rl^n\times\Rl^n,\Cl^{N\times N})$. The family of all $\Ccin$-kernel operators will be denoted $\W$.
\end{definition}

The relevant properties of $\Ccin$-kernel operators that we will use are the following.

\begin{lemma}\label{Lemma:CInftyIntegralOperator}{\bf (Properties of $\Ccin$-kernel operators)}\\
	Let $W$ be a $\Ccin$-kernel operator. Then 
	\begin{enumerate}
		\item[W1)] There exists a compact set $K\subset\Rl^n$ such that $W\Cin\subset\Ccin(K)$, and $Wf=0$ for all $f$ with $\supp f\cap K=\emptyset$,
		\item[W2)] $W$ extends to a continuous linear map $W:\Ltwon\to\Ccin(K)$.
		\item[W3)] The adjoint $W^*$ of $W$ w.r.t. $\langle\,\cdot\,,\,\cdot\,\rangle$ is also a $\Ccin$-kernel operator,
		\item[W4)] For any differential operator $Q$, also $WQ$ and $QW$ satisfy W1)--W3).
	\end{enumerate}
\end{lemma}
\begin{proof}
	{\em W1)} is clearly satisfied for any compact $K$ such that $\supp w\subset K\times K$. 
	
	{\em W2)} For compact $B\subset\Rl^n$, and every $a\in\Nl_0$, we find by a routine estimate
	\begin{align*}
		\sup_{x\in B\atop|\alpha|\leq a}|\partial_x^\alpha W_w f(x)|
		&\leq
		C_a\,\|f\|_2
		\,.
	\end{align*}
	Thus $W_w$ extends to a continuous map $W_w:\Ltwon\to\Ccin(K)$. Since ${W_w}^*=W_{w^*}$ with $w^*(x,y)=w(y,x)^*$, the same holds for the adjoint ${W_w}^*$, i.e. we have also shown {\em W3)}. Finally, acting with a differential operator $Q$ from the left on $W$ just results in a different $\Ccin$-kernel, as multiplication and differentiation preserve $\Ccin$. For the action from the right, $WQ$, one has to use integration by parts to arrive at the same conclusion. This shows {\em W4)}.
\end{proof}

The following investigations will be based on a (pre-)normally hyperbolic differential operator $D$ and a $\Ccin$-kernel operator $W$. By $K$, we will always refer to its ``support'', i.e. a compact subset of $\Rl^n$ as in {\em W1)}. Our requirements on $W$ can probably be relaxed (see also the remarks in Section~\ref{Section:StarProducts}), in particular regarding the smoothing property {\em W2)}. However, for the sake of simplicity, we stick to $\Ccin$-kernel operators for now.

\subsection{Fundamental Solutions and the Cauchy Problem}\label{Section:GreenAndCauchy}

Having fixed a (pre-)normally hyperbolic $D$ and a $\Ccin$-kernel operator $W\in\W$, we now consider, $\la\in\Cl$,
\begin{align}\label{eq:Dlambda}
	D_\la
	:=
	D+\la\,W\,,
\end{align}
which is defined as a continuous linear map $\Cin\to\Cin$. The first main step in our investigation will be the construction of advanced and retarded fundamental solutions for $D_\la$ (for small enough $|\la|$). These fundamental solutions will first be constructed in a suitable neighborhood of the support $K$ of the perturbation, and then on all of $\Rl^n$.

\begin{wrapfigure}{r}{0.36\textwidth}
  \begin{center}
    \includegraphics[width=0.33\textwidth]{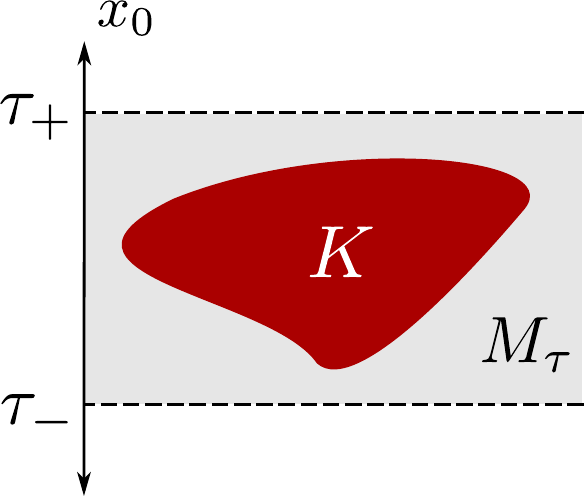}
  \end{center}
\end{wrapfigure}

To introduce this neighborhood, pick two numbers $\tau_-<\tau_+$ such that the time slice 
\begin{align*}
	M_\tau:=\{x\,:\,\tau_-<x_0<\tau_+\}=\Sigma_{\tau_-}^+\cap\Sigma_{\tau_+}^-
\end{align*}
contains $K$, where  $x=(x_0,...,x_{n-1})$ are the standard Cartesian coordinates of $\Rl^n$. This notation will be used throughout the article.

Restrictions to $M_\tau$ will generally be denoted by a subscript~$\tau$. For example, $D,D_\la$ naturally restrict to $\Cin(M_\tau)$, and these restrictions are denoted $D_\tau$, $D_{\tau,\la}$. For the fundamental solutions $R^\pm:\Ccin\to\Cin$, we denote the restriction in domain and range by $R^\pm_\tau$, i.e. $R_\tau^\pm:\Ccin(M_\tau)\to\Cin(M_\tau)$. As $Wf=0$ for $\supp f\cap M_\tau=\emptyset$, we omit the index $\tau$ when considering $W$ as restricted to $\Ltwon(M_\tau)$. Finally, the causal future/past of a set $B\subset M_\tau$ in $M_\tau$ is denoted $J_\tau^\pm(B)=J^\pm(B)\cap M_\tau$.
\\
\\
\indent As the time slice $M_\tau$ is causally convex,  {\em D7)} applies, the restricted differential operator $D_\tau$ has the unique advanced and retarded fundamental solutions $R_\tau^\pm$, i.e. $R_\tau^\pm:\Ccin(M_\tau) \to\Cin(M_\tau)$ are continuous and linear and satisfy $D_\tau R_\tau^\pm f=f=R_\tau^\pm D_\tau f$ and $\supp(R_\tau^\pm f)\subset J^\pm_\tau(\supp f)$ for any $f\in\Ccin(M_\tau)$. Due to the final extension of $M_\tau$ in time direction, we have the following additional statement. 

\begin{lemma}\label{lemma:RestrictedGreenOperators}
	The fundamental solutions $R_\tau^\pm$ of $D_\tau$ satisfy $R^\pm_\tau(\Ccin(M_\tau))\subset\Cin(M_\tau)\cap\Ltwon(M_\tau)$, and $R^\pm_\tau:\Ccin(M_\tau)\to\Ltwon(M_\tau)$ is continuous.
\end{lemma}
\begin{proof}
	Consider a compact set $B\subset M_\tau$. As $R_\tau^\pm:\Ccin(M_\tau)\to\Cin(M_\tau)$ is continuous, we have continuity of $R_\tau^\pm:\Ccin(B)\to \Cin(J^\pm_\tau(B))$. In view of the finite extension of $M_\tau$ in time direction, $J^\pm_\tau(B)$ is bounded, and furthermore, for any $f\in\Ccin(B)$, the smooth function $R_\tau^\pm f$ extends continuously to $\overline{M_\tau}$. This implies that we also have a continuous inclusion $R^\pm\Ccin(B)\xhookrightarrow{}\Ltwon(M_\tau)$. Thus $R_\tau^\pm(\Ccin(M_\tau))\subset\Cin(M_\tau)\cap\Ltwon(M_\tau)$, the map $R_\tau^\pm:\Ccin(B)\to\Ltwon(M_\tau)$ is continuous, and by definition of the inductive limit topology of $\Ccin(M_\tau)$, also the continuity $R_\tau^\pm:\Ccin(M_\tau)\to\Ltwon(M_\tau)$ follows.
\end{proof}

As a consequence of this lemma, we have the following result, which will be important in the sequel.

\begin{proposition}\label{proposition:WRbounded}
	The operators $WR^\pm_\tau$ and $R^\pm_\tau W$ extend from $\Ccin(M_\tau)$ to bounded operators on $\Ltwon(M_\tau)$. Furthermore, $R_\tau^\pm W(\Ltwon(M_\tau))\subset \Cin(M_\tau)\cap\Ltwon(M_\tau)$ and $WR_\tau^\pm(\Ltwon(M_\tau))\subset\Ccin(K)\subset\Ccin(M_\tau)$.
\end{proposition}
\begin{proof}
	By Lemma~\ref{lemma:RestrictedGreenOperators}, $R_\tau^\pm:\Ccin(M_\tau)\to\Ltwon(M_\tau)$ is continuous, with image contained in $\Cin(M_\tau)\cap\Ltwon(M_\tau)$, and by the smoothing property {\em W2)} of the perturbation, also $W:\Ltwon(M_\tau)\to\Ccin(M_\tau)$ is continuous. Hence the compositions 
	\begin{align*}
		R_\tau^\pm W:\Ltwon(M_\tau)\to\Ltwon(M_\tau)
		\,,\qquad
		WR_\tau^\pm:\Ccin(M_\tau)\to\Ccin(M_\tau)\,,	\end{align*}
	are well-defined and continuous, and the first one has image contained in $\Cin(M_\tau)\cap\Ltwon(M_\tau)$. Since continuity and boundedness are the same for linear maps on $\Ltwon(M_\tau)$, this already gives $\|R_\tau^\pm W\|<\infty$, where $\|\cdot\|$ denotes the norm of $\B(\Ltwon(M_\tau))$.

	To show the same for $WR^\pm_\tau$, recall that the adjoint differential operator $D_\tau^*$ also has continuous advanced and retarded fundamental solutions, which we denote here by $S^\pm_\tau$, and which are related to $R^\pm_\tau$ by $(S^\pm_\tau)^*=R^\mp_\tau$, cf. {\em D5)} and {\em D7)}. Taking also into account that $W^*\in\W$, it follows as above that also $S_\tau^\mp W^*$ extends to a bounded operator on $\Ltwon(M_\tau)$. Thus its adjoint $(S_\tau^\mp W^*)^*=W R_\tau^\pm$ is bounded as well.
	
	It remains to show $WR_\tau^\pm(\Ltwon(M_\tau))\subset\Ccin(K)$. To this end, let $\Delta=\partial_0^2+\partial_1^2+...+\partial_{n-1}^2$ denote the Laplace operator and recall that $(1-\Delta)^kW$ also extends to a bounded operator $\Ltwon(M_\tau)\to\Ccin(M_\tau)$, for any $k\in\Nl$ (Lemma~\ref{Lemma:CInftyIntegralOperator}~{\em W4)}). Thus $WR^\pm_\tau=(1-\Delta)^{-k}\cdot(1-\Delta)^kWR^\pm_\tau$ maps to smooth functions. Finally, since $Wf=0$ for $\supp f\cap K=\emptyset$, it is also clear that the image of $WR^\pm_\tau$ consists of functions with support in $K\subset M_\tau$, i.e. $WR_\tau^\pm(\Ltwon(M_\tau))\subset\Ccin(M_\tau)$.
\end{proof}

We shall now construct advanced and retarded fundamental solutions for $D_{\tau,\la}$ on the time slice $M_\tau$. As explained earlier, such Green's operators can be expected to be of the form $(1+\la R^\pm W)^{-1}R^\pm=\sum_{k=0}^\infty (-\la R^\pm W)^kR^\pm$. Thanks to the restriction to $M_\tau$, the convergence of the geometric series can be controlled. It turns out to be advantageous to also discuss the series with $R^\pm$ and $W$ interchanged, i.e., for $\lambda \in\Cl$, we introduce the series expressions
\begin{align} \label{eq:NeumannSeries}
	N^\pm_{\tau,\la} := \sum_{k=0}^\infty (-\lambda R^\pm_\tau W)^k
	\,,\qquad	
	{\tilde{N}}^\pm_{\tau,\la} = \sum_{k=0}^\infty (-\lambda WR^\pm_\tau)^k
	\,.
\end{align}
For these series, we can assert the following properties.
\newpage
\begin{proposition}\label{proposition:NeumannSeries}
	There exists $\la_0>0$ such that for all $\la\in\Cl$ with $|\la|<\la_0$,
	\begin{enumerate}
		\item the right hand sides of \eqref{eq:NeumannSeries} converge in the operator norm of the bounded linear operators on $\Ltwon(M_\tau)$, and therefore define bounded linear operators $N^\pm_{\tau,\la},\tilde{N}^\pm_{\tau,\la} : \Ltwon(M_\tau) \to \Ltwon(M_\tau)$,
		\item $N^\pm_{\tau,\la}$ and $\tilde{N}^\pm_{\tau,\la}$ are the inverse operators to $1+ \lambda R_\tau^\pm W$ and $1+ \lambda W R_\tau^\pm$, respectively (in the algebra of bounded  linear operators on $\Ltwon(M_\tau)$), i.e.,
		\begin{align}
			N^\pm_{\tau,\la} = (1 + \la R_\tau^\pm W)^{-1}
			\,,\qquad
			\tilde{N}^\pm_{\tau,\la} & = (1 + \la WR^\pm_\tau)^{-1}\,, 
		\end{align}
		\item $\tilde{N}^\pm_{\tau,\la}$ restricts to a continuous map $\tilde{N}^\pm_{\tau,\la}:\Ccin(M_\tau)\to\Ccin(M_\tau)$, and  $N^\pm_{\tau,\la}(\Cin(M_\tau)\cap\Ltwon(M_\tau))\subset\Cin(M_\tau)\cap\Ltwon(M_\tau)$.
		\item For $f\in\Ccin(M_\tau)$, we have $\supp(\tilde{N}^\pm_{\tau,\la}f)\subset\supp f\cup K$. If $\supp f\cap K=\emptyset$, then $N^\pm_{\tau,\la}f=f$.
		\item For $f\in\Ccin(M_\tau)$,
		\begin{align}\label{eq:Left=RightInverse}
			N^\pm_{\tau,\la}R_\tau^\pm f
			&=
			R_\tau^\pm \tilde{N}^\pm_{\tau,\la} f
			\,.
		\end{align}
	\end{enumerate}
\end{proposition}
\begin{proof}
	{\em a)} We have shown in Prop.~\ref{proposition:WRbounded} that $\lambda R_\tau^\pm W$ and $\lambda W R_\tau^\pm$ extend to bounded linear operators on $\Ltwon(M_\tau)$. We follow the usual practice and identify these operators with their bounded extensions. For $|\lambda|<\min\{\|R_\tau^\pm W\|^{-1},\|W R_\tau^\pm\|^{-1}\}=:\la_0$, the operator norms of $\lambda R_\tau^\pm W$, $\lambda W R_\tau^\pm$ are strictly smaller than 1, and hence the series on the right hand sides of \eqref{eq:NeumannSeries} converge in the operator norm. From now on, we only consider such $|\la|<\la_0$.
	\\[6pt]
	{\em b)} The series on the right hand sides of \eqref{eq:NeumannSeries} are Neumann series and thus coincide with $(1 + \la R_\tau^\pm W)^{-1}$ and $(1 + \la W R_\tau^\pm)^{-1}$, respectively (see e.g. \cite{BratteliRobinson:1987}, Sec.~2.2.1).
	\\[6pt]
	{\em c)} We have
	\begin{align}
		N^\pm_{\tau,\la}
		&=
		1 - \la R_\tau^\pm W \sum_{k=0}^\infty (-\la R_\tau^{\pm} W)^k
		=
		1-\la R_\tau^\pm W N_{\tau,\la}^\pm
		=
		1-\la  N_{\tau,\la}^\pm R_\tau^\pm W
		\label{eq:N}
		\,,\\
		\tilde{N}^\pm_{\tau,\la}
		&=
		1 - \la W R_\tau^\pm \sum_{k=0}^\infty (-\la W R_\tau^{\pm})^k
		=
		1-\la W R_\tau^\pm \tilde{N}_{\tau,\la}^\pm
		\,,
		\label{eq:Ntilde}
	\end{align}
	as operators on $\Ltwon(M_\tau)$. As $R_\tau^\pm W(\Ltwon(M_\tau))\subset\Cin(M_\tau)\cap\Ltwon(M_\tau)$ and $WR_\tau^\pm(\Ltwon(M_\tau))\subset\Ccin(M_\tau)$ (Prop.~\ref{proposition:WRbounded}), the claimed restrictions of
	$N^\pm_{\tau,\la}$ and $\tilde{N}^\pm_{\tau,\la}$ follow. Furthermore, since $\Ccin(M_\tau)$ is continuously embedded in $\Ltwon(M_\tau)$, $\tilde{N}^\pm_{\tau,\la}$ is a bounded operator on $\Ltwon(M_\tau)$, and $WR^\pm_\tau:\Ltwon(M_\tau)\to\Ccin(M_\tau)$ is continuous, also the continuity of $\tilde{N}^\pm_{\tau,\la}:\Ccin(M_\tau)\to\Ccin(M_\tau)$ follows.
	\\[6pt]
	{\em d)} This follows immediately from (\ref{eq:N}, \ref{eq:Ntilde}) and the support properties {\em W1)}.
	\\[6pt]
	{\em e)} We first note that by Lemma~\ref{lemma:RestrictedGreenOperators} and part {\em c)}, the expressions $N^\pm_{\tau,\la}R_\tau^\pm f$ and $R_\tau^\pm \tilde{N}^\pm_{\tau,\la} f$ are well-defined for $f\in\Ccin(M_\tau)$. Then, with arbitrary $g\in\Ccin(M_\tau)$,
	\begin{align*}
		\langle g,\,N^\pm_{\tau,\la}R^\pm_\tau f\rangle
		&=
		\sum_{k=0}^\infty \langle g,\,(-\la R^\pm_\tau W)^k R^\pm_\tau f\rangle
		=
		\sum_{k=0}^\infty \langle g,\,R^\pm_\tau (-\la W R^\pm_\tau )^k f\rangle
		\\
		&=
		\sum_{k=0}^\infty \langle(R^\pm_\tau)^*g,\,(-\la W R^\pm_\tau )^k f\rangle
		=
		\langle (R^\pm_\tau)^*g,\,\tilde{N}^\pm_{\tau,\la} f\rangle
		=
		\langle g,\,R^\pm_\tau\tilde{N}^\pm_{\tau,\la} f\rangle
		\,.
	\end{align*}
	This implies \eqref{eq:Left=RightInverse}.
\end{proof}

Proposition~\ref{proposition:NeumannSeries} puts us in the position to obtain fundamental solutions of $D_{\tau,\la}$. Here and in the following, we only consider $\la$ with $|\la|<\min\{\|R^+_\tau W\|^{-1},\|R^-_\tau W\|^{-1},\|WR^+_\tau\|^{-1},\|WR^-_\tau\|^{-1}\}$, so that we can use the preceding results, and indicate that by writing ``for sufficiently small~$|\la|$''. 
\\
\\
\indent As an aside, we mention that this restriction on the coupling $\la$ can also be understood as a way of preserving the hyperbolic character of $D$. In fact, a perturbation $W\in\W$ with general coupling $\la$ can change the hyperbolic character of $D$ drastically, for example to the effect that there exist solutions of the homogeneous equation $D_\la f=0$ that have compact support.

\begin{example}{\bf (Compactly supported solutions)}\\
	Let $W$ be a $\Ccin$-kernel operator of the form $Wf=\langle w_1,f\rangle\cdot Dw_2$, with $w_1,w_2\in\Ccin$ such that $\langle w_1,w_2\rangle\neq0$. Then there exists $\la\in\Cl$ such that $D_\la$ has non-zero compactly supported solutions.
\end{example}
\begin{proof}
	One computes $(D+\la W)f=Df+\la\langle w_1,f\rangle\cdot Dw_2$, and this expression vanishes for $f=w_2\in\Ccin$ and $\la=-\langle w_1,w_2\rangle^{-1}$.
\end{proof}

If compactly supported solutions exist, there can be no unique fundamental solutions, and also quantization will be ambiguous. However, these compactly supported solutions do not exist for sufficiently small $|\la|$.

\begin{lemma}\label{lemma:NoCompactlySupportedSolutions}
	Let $|\la|$ be sufficiently small, and $f\in\Ccin(M_\tau)$. If $D_{\tau,\la}f=0$, then $f=0$.
\end{lemma}
\begin{proof}
	By assumption, we have $D_\tau f=-\la Wf$, with $f\in\Ccin(M_\tau)$. Applying $R^\pm_\tau$ therefore gives $R^\pm_\tau D_\tau f=f=-\la R^\pm_\tau Wf$, i.e. either the $\Ltwon(M_\tau)$-operator $-\la R^\pm_\tau W$ has the eigenvalue~$1$, or $f=0$. But we fixed $\la$ in such a way that $\|-\la R^\pm_\tau W\|<1$. Hence $f=0$.
\end{proof}

After this remark, we proceed to the fundamental solutions of $D_{\tau,\la}$, and introduce the operators 
\begin{align} \label{eq:GreenOperators}
	R_{\tau,\lambda}^\pm & := N^\pm_{\tau,\la}R_\tau^\pm = R_\tau^\pm \tilde{N}^\pm_{\tau,\la}: 
	\Ccin(M_\tau) \to \Cin(M_\tau) \cap \Ltwon(M_\tau)
	\,,
\end{align}
which are well defined by the properties of $R_\tau^\pm$ (Lemma~\ref{lemma:RestrictedGreenOperators}) and $N_{\tau,\la}^\pm, \tilde{N}_{\tau,\la}^\pm$  (Prop.~\ref{proposition:NeumannSeries}). By Prop.~\ref{proposition:NeumannSeries}~{\em c)}, they are also continuous as maps $\Ccin(M_\tau)\to\Cin(M_\tau)$.

\begin{theorem}\label{theorem:GreenOnTau} {\bf (Fundamental solutions on a time slice)}\\
	For sufficiently small $|\la|$, the operators $R_{\tau,\la}^\pm:\Ccin(M_\tau)\to\Cin(M_\tau)$ \eqref{eq:GreenOperators} exist as continuous linear maps and satisfy, $f,g\in\Ccin(M_\tau)$
	\begin{enumerate}
		\item $D_{\tau,\la}R^\pm_{\tau,\la}f=f=R^\pm_{\tau,\la}D_{\tau,\la}f$.
		\item $\supp(R^\pm_{\tau,\la} f) \subset 		J_\tau^\pm(\supp f) \cup J_\tau^\pm (K)$.
		\item $\supp(R^\pm_{\tau,\la}f-R^\pm_\tau f)\subset J_\tau^\pm(K)$.
		\item If $J^\pm_\tau(\supp f)\cap K=\emptyset$, then 
			\begin{align}\label{eq:FDoesntHitK}
				R^\pm_{\tau,\la}f
				=
				R^\pm_\tau f
				\,.
			\end{align}
		\item If $D$ and $W$ are symmetric, i.e. $D=D^*$, $W=W^*$, and $\la\in\Rl$, then one has, $f,g \in \Ccin(M_\tau)$,
			\begin{align}
				\langle g, R_{\tau,\la}^\pm f\rangle
				=
				\langle R^{\mp}_{\tau,\la}g,f\rangle
				\,.
			\end{align}
		\item If $\tau=(\tau_-,\tau_+)$ is replaced by $\tau'=(\tau_-',\tau_+')$, with $\tau_-'>\tau_-$ and $\tau_+'<\tau_+$ such that $K\subset M_{\tau'}\subset M_\tau$, the statements a)--e) still hold.
	\end{enumerate}
\end{theorem}
\begin{proof} 
	{\em a)} Note that the perturbed differential operator $D_{\tau,\la}$, which is defined on $\Cin(M_\tau)$, restricts to $\Ccin(M_\tau)$ by the properties of $D_\tau$ and $W$. Hence both compositions, $D_{\tau,\la}R_{\tau,\la}^\pm$ and $R_{\tau,\la}^\pm D_{\tau,\la}$, are well-defined on $\Ccin(M_\tau)$.
	
	For $f \in \Ccin(M_\tau)$, we compute
	\begin{align*}
		R_{\tau,\lambda}^\pm D_{\tau,\la}f
		= 
		N_{\tau,\la}^\pm R_\tau^\pm (D_\tau+\la\,W)f
		=
		(1+\la R_\tau^\pm W)^{-1}\,(f+\la R_\tau^\pm Wf)
		=
		f\,,
	\end{align*}
	where we have used that $R_\tau^\pm$ is a fundamental solution of $D_\tau$, i.e. $R_\tau^\pm D_\tau f=f$. Similarly,
	\begin{align*}
		D_{\tau,\la} R_{\tau,\lambda}^\pm f
		&= 
		(D_\tau+\la\,W) R_\tau^\pm \tilde{N}_{\tau,\la}^\pm f
		=
		\tilde{N}_{\tau,\la}^\pm f+\la W R_\tau^\pm \tilde{N}_{\tau,\la}^\pm f
		\\
		&=
		(1+\la WR^\pm_\tau)(1+\la WR^\pm_\tau)^{-1}f
		=
		f.
	\end{align*}
	{\em b), c), d)}: Using the hyperbolic character of $R^\pm_\tau$ and Prop.~\ref{proposition:NeumannSeries}~{\em d)}, we get
	\begin{align*}
		\supp(R^\pm_{\tau,\la} f)
		=
		\supp(R_\tau^\pm \tilde{N}^\pm_{\tau,\la} f)
		\subset
		J_\tau^\pm(\supp f \cup K)
		=
		J_\tau^\pm(\supp f) \cup J_\tau^\pm(K) \,.
	\end{align*}
	Replacing $f$ by $R^\pm_{\tau,\la}f$ in the second statement of  Prop.~\ref{proposition:NeumannSeries}~{\em d)},  we also get $R^\pm_{\tau,\la}f=N^\pm_{\tau,\la}R^\pm_\tau f=R^\pm_\tau f$ in case $\supp(R^\pm_\tau f)\subset J^\pm_\tau(\supp f)$ is disjoint from $K$, i.e. eqn.~\eqref{eq:FDoesntHitK}. For {\em c)}, observe that by \eqref{eq:Ntilde}
	\begin{align*}
		R^\pm_\tau f-R^\pm_{\tau,\la}f
		=
		R^\pm_\tau(1-\tilde{N}^\pm_{\tau,\la})f
		=
		\la\,R^\pm_\tau\,WR^\pm_\tau \tilde{N}^\pm_{\tau,\la}f
		\,.
	\end{align*}
	As $WR^\pm_\tau \tilde{N}^\pm_{\tau,\la}f$ has support in $K$, we get
	$\supp(R^\pm_\tau f-R^\pm_{\tau,\la}f)\subset J_\tau^\pm(K)$ as claimed.
	
	{\em e)} Using the symmetry of $D$, and thus of $D_\tau$, we have $(R^\pm_\tau)^*=R^\mp_\tau$. With $W=W^*$, this gives
	\begin{align*}
		\langle g,R^\pm_{\tau,\la},f\rangle
		&=
		\sum_{k=0}^\infty \langle g,R^\pm_\tau(-\la W R^\pm_\tau)^k f\rangle
		=
		\sum_{k=0}^\infty \langle(-\lambda R^\mp_\tau W)^k R^\mp_\tau g,f\rangle
		=
		\langle R^\mp_{\tau,\lambda} g, f\rangle
		\,.
	\end{align*}
	{\em f)} The operator norms of the restrictions of the bounded operators $R_\tau^\pm W$, $WR_\tau^\pm$ from $\Ltwon(M_\tau)$ to the subspace $\Ltwon(M_{\tau'})$ are not larger than the norms of their unrestricted counterparts. Thus any $|\la|$ that is sufficiently small for the time cutoff $\tau$ is also sufficiently small for the sharper time cutoff $\tau'$, and hence {\em a)--e)} remain valid for $\tau'$ instead of $\tau$.
\end{proof}

This theorem shows that the $R^\pm_{\tau,\la}$ are quite close to advanced/retarded fundamental solutions, with possible acausal propagation in the future/past of the perturbation region $K$. Despite these differences to advanced/retarded fundamental solutions for local differential operators, we will refer to the $R^\pm_{\tau,\la}$ with the same terminology as in the local case.

Defining $R_{\tau,\la}:=R^-_{\tau,\la}-R^+_{\tau,\la}$, it is clear from part {\em a)} of the theorem that any function of the form $R_{\tau,\la}g$, $g\in\Ccin(M_\tau)$, is a solution of $D_{\tau,\la}$.
\\
\\
In the next step we will extend the fundamental solutions in the time slice $M_\tau$ to all of $\Rl^n$. As the potential $W$ vanishes outside $M_\tau$, this amounts to a ``gluing'' of advanced/retarded solutions of $Df=0$ outside $M_\tau$ with advanced/retarded solutions of $D_{\tau,\la}f=0$ in $M_\tau$.

\begin{wrapfigure}{r}{0.40\textwidth}
  \begin{center}
    \includegraphics[width=0.38\textwidth]{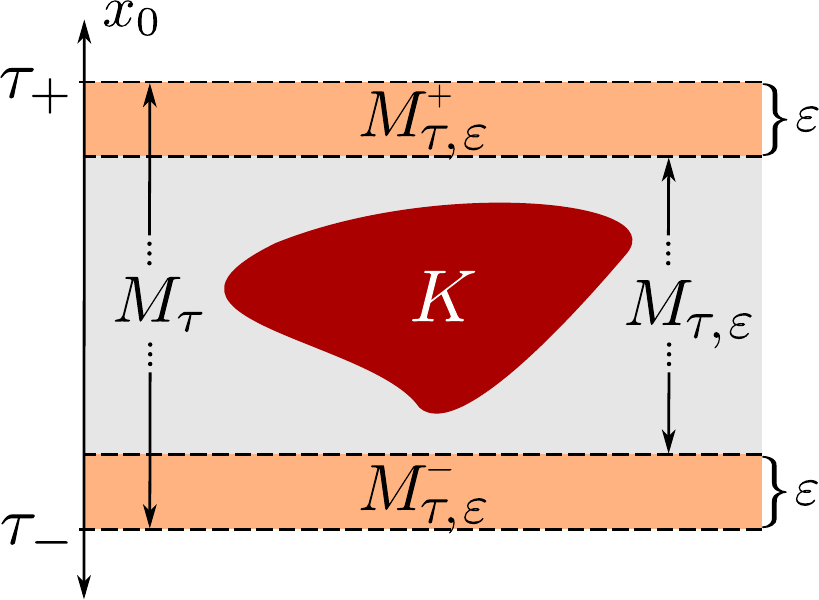}
  \end{center}
\end{wrapfigure}

We first introduce some notation. Let $\eps>0$, and define
\begin{align*}
	M_{\tau,\eps}^-
	&:=
	\Sigma_{\tau_-}^+\cap\Sigma_{\tau_-+\eps}^-
	\,,\\
	M_{\tau,\eps}^+
	&:=
	\Sigma_{\tau_+-\eps}^+\cap\Sigma_{\tau_+}^-
	\,,\\
	M_{\tau,\eps}
	&:=
	\Sigma_{\tau_-+\eps}^+\cap\Sigma_{\tau_+-\eps}^-
	\,,
\end{align*}
We require that $\eps$ is so small that $K\subset M_{\tau,\eps}$, as depicted in the figure on the right.

Given $h\in\Ccin(M_{\tau,\eps})$, the function $R^+_{\tau,\la}h$ vanishes on $M_{\tau,\eps}^-$, and is a solution of $D$ with compactly supported Cauchy data on $M^+_{\tau,\eps}$ --- this follows from Thm.~\ref{theorem:GreenOnTau}~{\em b)} and the fact that $K$ does not intersect $M^\pm_{\tau,\eps}$. We denote its Cauchy data on $\Sigma_{\tau_+-\eps}$ by $u_h^+$. Similarly, $R^-_{\tau,\la}h$ vanishes on $M_{\tau,\eps}^+$, and is a solution of $D$ with compactly supported Cauchy data on $M^-_{\tau,\eps}$; its Cauchy data on $\Sigma_{\tau_-+\eps}$ will be denoted $u_h^-$.

Given Cauchy data $u$ on some Cauchy hyperplane $\Sigma$, the corresponding solution of $D$ will always be denoted $f_0[u]$. We define
\begin{align}\label{eq:RLambda1}
	(R_\la^\pm h)(x)
	:=
	\begin{cases}
		(R_{\tau,\la}^\pm h)(x)  & x\in M_\tau\\ 
		f_0[u_h^\pm](x) & x\in\Sigma_{\tau_\pm\mp\eps}^\pm\\
		0 & x\in\Sigma_{\tau_\mp\pm\eps}^\mp
	\end{cases}
	\,,\qquad
	h\in\Ccin(M_{\tau,\eps})
	\,.
\end{align}
This assignment is well-defined in the overlap regions $M^+_{\tau,\eps}$ and $M^-_{\tau,\eps}$. In fact, $(R_{\tau,\la}^\pm h)(x)=0$ for $x\in M^\mp_{\tau,\eps}$ as recalled above, and $R_{\tau,\la}h$ is a solution of $D$ on the strip $M^\pm_{\tau,\eps}$. As this solution is uniquely fixed by its Cauchy data, it coincides with $f_0[u_h^\pm]$ in this region.

It is also clear that \eqref{eq:RLambda1} restricts to $R^\pm_{\tau,\la}h$ on $M_\tau$, and is an advanced/retarded fundamental solution of $D_\la$ in the sense that $D_\la R^\pm_\la h=h=R^\pm_\la D_\la h$, and $\supp(R_\la^\pm h)\subset J^\pm(\supp h)\cup J^\pm(K)$ --- the latter statement is a consequence of Thm.~\ref{theorem:GreenOnTau}~{\em b)} and {\em D4)}. Also the items {\em c)--d)} of Thm.~\ref{theorem:GreenOnTau} hold for $R^\pm_\la h$ when the index $\tau$ is dropped and $M_\tau$ is replaced by $\Rl^n$ throughout. 

In a similar fashion, we now want to define $R_\la^\pm$ on functions $h\in\Ccin$ whose support lies outside of $M_\tau$. For $\supp h\subset \Sigma^\pm_{\tau_\pm}$, the function $R^\pm h$ vanishes on $M_\tau$, and we therefore simply set
\begin{align}
	R_\la^\pm h
	:=
	R^\pm h
	\,,\qquad
	h\in\Ccin(\Sigma^\pm_{\tau_\pm})
	\,.
\end{align}
\begin{wrapfigure}{r}{0.40\textwidth}
  \begin{center}
    \includegraphics[width=0.38\textwidth]{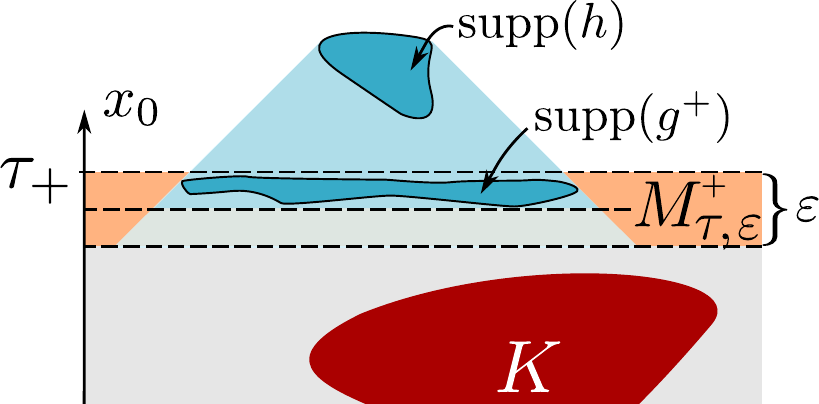}
  \end{center}
\end{wrapfigure}
To define $R_\la^\mp h$, $h\in\Ccin(\Sigma^\pm_{\tau_\pm})$, we observe that $R^\mp h$ is a solution of $D$ in the strip $M^\pm_{\tau,\eps}$. This solution has compactly supported Cauchy data, and according to {\em D6)}, we may therefore represent it in the form $(R^\mp h)(x)=\pm(Rg^\pm)(x)$, $x\in M^\pm_{\tau,\eps}$, for some $g^\pm\in\Ccin$ which is supported in the strip of half the width, $\supp(g^\pm)\subset M^\pm_{\tau,\eps/2}$. On the inner half of $M^\pm_{\tau,\eps}$, we then have $(R^\mp h)(x)=\pm(Rg^\pm)(x)=(R^\mp g^\pm)(x)=(R^\mp_{\tau,\la}g^\pm)(x)$. This implies that 
\begin{align}
	(R_\la^\mp h)(x)
	:=
	\begin{cases}
		(R^\mp h)(x) & x\in \Sigma^\pm_{\tau_\pm \mp\eps/2}
		\\
		(R^\mp_{\tau,\la}g^\pm)(x) & x\in M_\tau
	\end{cases}
	\,,\qquad 
	h\in\Ccin(\Sigma^\pm_{\tau_\pm})\,.
\end{align}
is well-defined. It remains to define $(R_\la^\mp h)(x)$ for $x\in\Sigma^\mp_{\tau_\mp}$. To do so, we proceed as in the definition of $R_\la^\pm h$ for $h\in\Ccin(M_\tau)$, and let $v$ denote the Cauchy data of $R^\mp_{\tau,\la}g^\pm$ on the Cauchy hyperplane $\Sigma_{\tau_\mp\pm\eps}$. Then we set
\begin{align}
	(R_\la^\mp h)(x)
	:=
	f_0[v](x)
	\,,\qquad
	x\in \Sigma^\mp_{\tau_\mp\pm\eps}
	,\;h\in\Ccin(\Sigma^\pm_{\tau_\pm})\,.
\end{align}
As before, the assignment is well-defined in the overlap region, and completes our definition of $R_\la^\mp h$. By construction, it is clear that again the statements Thm.~\ref{theorem:GreenOnTau} {\em a)--d)} hold for $R^\pm_\la h$ when the index $\tau$ is dropped and $M_\tau$ is replaced by $\Rl^n$.

Making use of Thm.~\ref{theorem:GreenOnTau} {\em e)}, it also becomes apparent that our construction is independent of $\eps$, and also results in the same definition of $R^\pm_\la h$ when $\tau$ is replaced by a sharper cut-off $\tau'$ such that $\tau_+'<\tau_+$, $\tau_-'>\tau_-$, $K\subset M_{\tau'}$. We can thus proceed to the definition of $R_\la^\pm h$ for $h\in\Ccin$ of arbitrary support with the help of a smooth partition of unity. In fact, let $1=\chi_++\chi_0+\chi_-$ be a smooth partition of unity, where $\chi_\pm,\chi_0$ are smooth functions on $\Rl$ with supports $\supp\chi_+\subset(\tau_+-\eps,\infty)$, $\supp\chi_0\subset(\tau_-,\tau_+)$, $\supp\chi_-\subset(-\infty,\tau_-+\eps)$. Denoting the multiplication operators with $\chi_\pm(x_0)$, $\chi_0(x_0)$ by the same letters, we then set
\begin{align}
	R_\la^\pm h
	:=
	R_\la^\pm\chi_+h + R_\la^\pm\chi_0h+R_\la^\pm\chi_-h
	\,,\qquad
	h\in\Ccin .
\end{align}
All functions on the right hand side have been defined before, and the left hand side inherits properties {\em a)--d)} of Thm.~\ref{theorem:GreenOnTau} from them. Finally, also Thm.~\ref{theorem:GreenOnTau}~{\em e)} transports to the global case: For $D=D^*$, $W=W^*$, $\la\in\Rl$, the integral $\langle f,R_\la^\pm g\rangle$, with $f,g\in\Ccin$, can be split in two parts, namely one integral over $M_\tau$ and one integral over $\Rl^n\backslash M_\tau$. On $M_\tau$, $R_\la^\pm$ restrict to $R_{\tau,\la}^\pm$, and we can use Thm.~\ref{theorem:GreenOnTau}~{\em e)} to compute the adjoint. On the complement, the same conclusion follows from exploiting the properties of $R^\pm$.

We summarize the results of our construction in the following theorem.

\begin{theorem} {\bf (Global fundamental solutions)}\label{theorem:GlobaleFundamentalSolutions}\\
	For sufficiently small $|\la|$, the operators $R_\la^\pm:\Ccin\to\Cin$ defined above exist as continuous linear maps and satisfy, $f,g\in\Ccin$,
	\begin{enumerate}
		\item $D_{\la}R^\pm_{\la}f=f=R^\pm_{\la}D_{\la}f$.
		\item $\supp(R^\pm_{\la} f) \subset 		J^\pm(\supp f) \cup J^\pm (K)$.
		\item $\supp(R^\pm_\la f-R^\pm f)\subset J^\pm(K)$.
		\item If $J^\pm(\supp f)\cap K=\emptyset$, then $R^\pm_\la f=R^\pm f$.
		\item If $D$ and $W$ are symmetric, i.e. $D=D^*$, $W=W^*$, and $\la\in\Rl$, then one has, $f,g \in \Ccin$,
			\begin{align}
				\langle g, R_\la^\pm f\rangle
				=
				\langle R^{\mp}_\la g,f\rangle
				\,.
			\end{align}
		{\hfill $\square$ \\[2mm] \indent}
	\end{enumerate}
\end{theorem}

\begin{figure}[h]
\begin{center}
	\includegraphics[width=115mm]{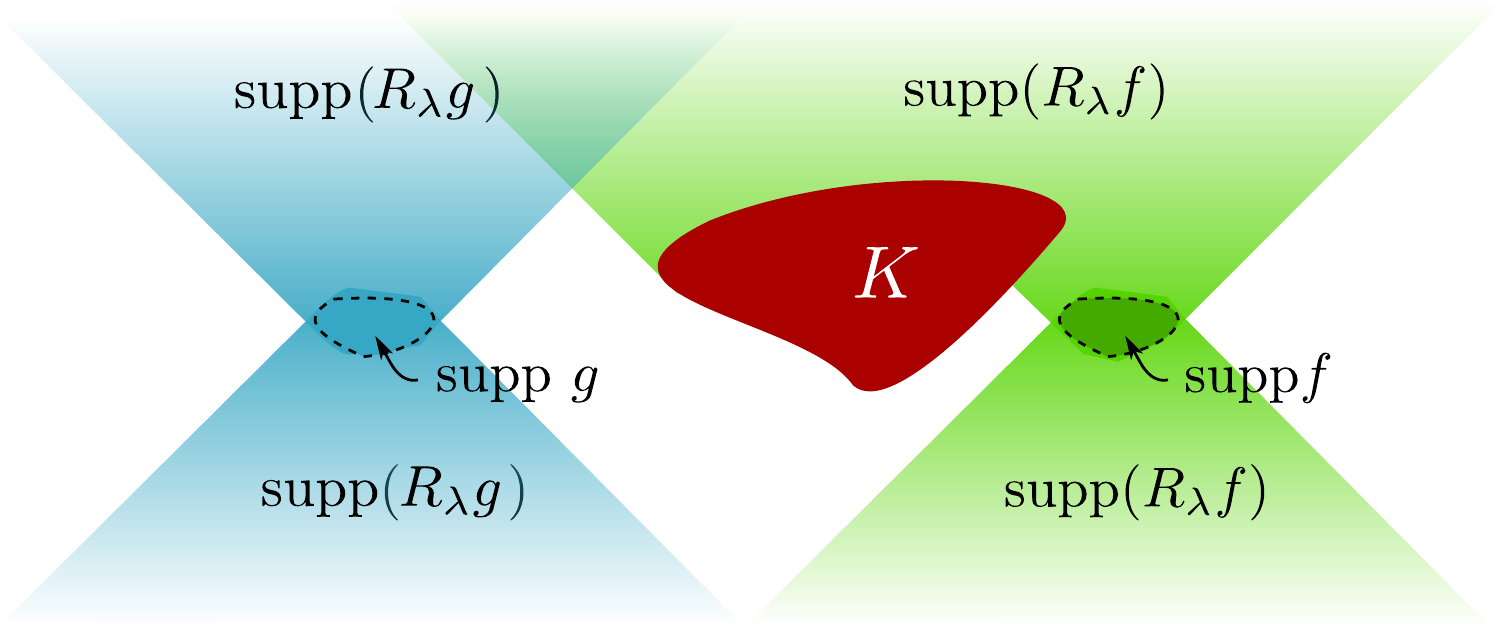}	
	\\{\em Typical supports of two solutions $R_\la f$, $R_\la g$ of $D_\la$.}
\end{center}
\end{figure}

According to Theorem~\ref{theorem:GlobaleFundamentalSolutions}, the influence of the perturbation $W$ is confined to the future/past of its support $K$, as if a source $u$ with support in $K$ would have been added to the unperturbed equation, i.e. $Df=u$. However, in contrast to the solutions of $Df=u$, the solutions of $D_\la f=0$ do not differ from the ones of $Df=0$ in case $J^\pm(\supp f)\cap K=\emptyset$, i.e. if the unperturbed wave does not collide with the potential in $K$, as depicted in the figure below.
\\
\\
\indent So far we made no claim concerning uniqueness of the fundamental solutions, and in fact, a couple of choices were made in the construction of the $R^\pm_\la$. However, we will show below that the $R^\pm_\la$ are in fact uniquely determined by properties {\em a)} and {\em b)} of the preceding theorem, and in particular independent of the choices made.

\begin{proposition}\label{proposition:Uniqueness}{\bf(Uniqueness properties)}
	\begin{enumerate}
		\item Let $f\in\Cin$ satisfy $D_\la f=0$ and $\supp f\subset V^++a$ or $\supp f\subset V^-+a$ for some $a\in\Rl^n$. Then $f=0$.
		\item 	The advanced/retarded fundamental solutions $R_\la^\pm$ are the unique linear maps $\Ccin\to\Cin$ satisfying properties a) and b) of Theorem~\ref{theorem:GlobaleFundamentalSolutions}.
	\end{enumerate}
\end{proposition}
\begin{proof}
	{\em a)} This argument is based on Green's second identity (or Gauss' theorem), and we will have to distinguish the cases where $D$ is normally hyperbolic and prenormally hyperbolic, respectively. We begin with the normally hyperbolic case, i.e. $D=\square+U^\mu(x)\partial_\mu+V(x)$, where $U^\mu, V:\Rl^n\to\Cl^{N\times N}$ are smooth functions.
	
	Let $f$ be a solution of $D_\la$, $h\in\Ccin$, and $\Sigma_t$ a Cauchy hyperplane. We set $g_\la^\pm:=T_\la^\mp h$, where $T^\mp_\la=(R^\pm_\la)^*$ are the advanced/retarded fundamental solutions of $D_\la^*=D^*+\overline{\la}W^*$, so that the support of $x\mapsto(g^\pm(x),\partial_\mu f(x))$ has compact intersection with $\Sigma^\pm$. We can thus use Green's second identity and integration by parts to compute
	\begin{align*}
		\int_{\Sigma_t^\pm}\left((D^*g^\pm,f)-(g^\pm,Df)\right)
		&=
		\int_{\Sigma_t^\pm}\left((\square g^\pm,f)-(g^\pm,\square f)\right)
		-
		\int_{\Sigma_t^\pm}\left((\partial_\mu\overline{U^\mu} g^\pm,f)-(g^\pm,U^\mu\partial_\mu f)\right)
		\\
		&=
		\mp\int_{\Sigma_t}\left\{(\partial_0 g^\pm_t,f_t)-(g^\pm_t,\partial_0 f_t)\right\}
		\pm\int_{\Sigma_t}(g_t^\pm,(U^0f)_t)
		\,,
	\end{align*}
	where an index $t$ denotes restriction to $\Sigma_t$. The left hand side can also be evaluated by using the equations $Df=-\la Wf$ (since $f$ is a solution) and $D^*g^\pm=h-\overline{\la} W^*g^\pm$ (by definition of $g^\pm$),	
	\begin{align*}
		\int_{\Sigma_t^\pm}\left((D^*g^\pm,f)-(g^\pm,Df)\right)
		&=
		\int_{\Sigma_t^\pm}(h,f)-\la\int_{\Sigma_t^\pm}(W^*T_\la^\mp h,f)+\la\int_{\Sigma_t^\pm}(T_\la^\mp h,Wf)
		\,.
	\end{align*}
	 Adding the equations for both choices of ``$\pm$'' then gives
	\begin{align}
		\langle h,f\rangle
		&-
		\la\left\{
			\int_{\Sigma_t^+}(W^*T_\la^-h,f)
			+
			\int_{\Sigma_t^-}(W^*T_\la^+h,f)
			-
			\int_{\Sigma_t^+}(T_\la^-h,Wf)
			-
			\int_{\Sigma_t^-}(T_\la^+h,Wf)
		\right\}
		\nonumber
		\\
		&=
		\int_{\Sigma_t}
		\bigg\{
			((\partial_0 R^*_\la h)_t,f_t)
			-
			((R^*_\la h)_t,(\partial_0 f)_t)
		\bigg\}
		-\int_{\Sigma_t}
		((R^*_\la h)_t,U^0_t f_t)
		\label{eq:GaussMitRandtermen}
		\,.
	\end{align}
	Suppose now that $\supp f\subset V^++a$ or $\supp f\subset V^-+a$ for some $a\in\Rl^n$. Then we can choose $\Sigma_t$ in such a way that $K\subset\Sigma_t^\pm$ and $f_t=0$, $(\partial_0 f)_t=0$. In this situation, the right hand side of the above equation vanishes, and the four terms in curly brackets cancel because in each of these integrals, the range of integration can be taken as $\Rl^n$ instead of $\Sigma_t^\pm$. Hence we arrive at $\langle h,f\rangle=0$. As $h\in\Ccin$ was arbitrary, this implies $f=0$.
	\\
	\\
	For the case that $D$ is prenormally hyperbolic, we find another prenormally hyperbolic $D'$ such that $D'D$ is again normally hyperbolic. Moreover, by Lemma~\ref{Lemma:CInftyIntegralOperator}~{\em W4)}, also $D'W\in\W$. Thus, if $f_\la$ is a solution of $D_\la$, then $0=D'D_\la f_\la=(D'D+\la D'W)f_\la$, i.e. $f_\la$ is also a solution of the normally hyperbolic operator $D'D$, perturbed by $\la D'W\in\W$. By our previous argument for normally hyperbolic operators, we then see that $f_\la=0$ if $f_\la$ has support in a light cone.

	{\em b)} If $\tilde{R}^\pm_\la$ is another linear map satisfying Thm.~\ref{theorem:GlobaleFundamentalSolutions} {\em a),b)}, then for any $f\in\Ccin$, the function $R^\pm_\la f-\tilde{R}^\pm_\la f$ is a solution of $D_\la$ (because of Thm.~\ref{theorem:GlobaleFundamentalSolutions} {\em a)}) with support in a future/past light cone (because of Thm.~\ref{theorem:GlobaleFundamentalSolutions} {\em b)}). Hence $R^\pm_\la f=\tilde{R}^\pm_\la f$ by part {\em a)}.
\end{proof}

Having established the basic existence and uniqueness theorem on fundamental solutions, we introduce in complete analogy to the unperturbed case the space of all solutions of $D_\la$ with compactly supported Cauchy data as
\begin{align}\label{eq:SolutionSpace}
	\Sol_\la
	:=
	\{f\in\Cin\,:\,D_\la f=0\,,\;\supp f \subset (V^++a_+)\cup(V^-+a_-)\,\text{ for some } a_\pm\in\Rl^n\}
	\,,
\end{align}
and define the propagator as
\begin{align}
	R_\la
	:=
	R^-_\la-R^+_\la
	\,.
\end{align}

\begin{proposition}\label{proposition:SolutionSpace}{\bf(Structure of the solution spaces)}
	\begin{enumerate}
		\item Let $\bSigma$ denote an open causally convex neighborhood of a Cauchy hyperplane $\Sigma$ such that $K\subset J^+(\bSigma)\backslash\bSigma$ or $K\subset J^-(\bSigma)\backslash\bSigma$. Then
			\begin{align}
				\Sol_\la=R_\la\Ccin(\bSigma)\,.
			\end{align}
		\item We have $\ker R_\la=D_\la\Ccin$, and hence $\Sol_\la \cong \Ccin/\ker R_\la=\Ccin/D_\la \Ccin$.
		\item If $D=D^*$, $W=W^*$, and $\la\in\Rl$, the map 
		\begin{align}
			\rho_\la : \Sol_\la\times \Sol_\la &\to\Cl\,,\\
			(R_\la f,R_\la g)
			&\mapsto
			\langle f,R_\la g\rangle
			\label{eq:Rho}
		\end{align}
		is a well-defined non-degenerate sesquilinear form satisfying
		\begin{align}\label{eq:RhoSymmetry}
			\overline{\rho_\la(R_\la f,R_\la g)}
			=
			-\rho_\la(R_\la g,R_\la f)
			\,,\qquad
			f,g\in\Ccin\,.
		\end{align}
	\end{enumerate}
\end{proposition}
\begin{proof}
	{\em a)} We carry out the proof for the case $K\subset J^+(\bSigma)\backslash\bSigma$, the other case is analogous. Let $f_\la\in\Sol_\la$ and consider the restriction $f_\la|_{\bSigma}$. As $\bSigma$ is disjoint from $K$, this restriction is a solution of $D$ on $\bSigma$, and thus there exists $g\in\Ccin(\bSigma)$ such that $f_\la|_{\bSigma}=(Rg)|_{\bSigma}=(R_\la g)|_{\bSigma}$. Hence the two solutions $f_\la$, $R_\la g$ of $D_\la$ coincide on $\bSigma$, i.e. $f_\la-R_\la g=h^++h^-$, where $h^\pm\in\Cin$ have support in $J^\pm(\bSigma)\backslash\bSigma$. As $K\subset J^+(\bSigma)\backslash\bSigma$, both, $h^+$ and $h^-$, are solutions of $D_\la$, and in view of the support properties of $f_\la$ (see \eqref{eq:SolutionSpace}) and $R_\la g$ (see Thm.~\ref{theorem:GlobaleFundamentalSolutions} {\em b)}), we have $\supp h^\pm\subset V^\pm+a_\pm$ for some $a_\pm\in\Rl^n$. Thus, by Proposition~\ref{proposition:Uniqueness} {\em a)}, $h^+=h^-=0$, and $f_\la=R_\la g$.
	
	{\em b)} By Theorem~\ref{theorem:GlobaleFundamentalSolutions} {\em a)}, we have $R_\la D_\la f=0$ for any $f\in\Ccin$, i.e. $D_\la\Ccin\subset\ker R_\la$. Conversely, for $f\in\ker R_\la$, the function $g:=R_\la^+f=R_\la^-f$ has compact support in view of Theorem~\ref{theorem:GlobaleFundamentalSolutions}~{\em b)}. Thus $f=D_\la R_\la^+ f=D_\la g\in D_\la\Ccin$, i.e. we have shown $\ker R_\la=D_\la\Ccin$. Now, by part {\em a)}, we know $\Sol_\la=R_\la\Ccin$, and thus $\Sol_\la \cong \Ccin/\ker R_\la=\Ccin/D_\la \Ccin$.
	
	{\em c)} For $f,g\in\Ccin$, we have by Theorem~\ref{theorem:GlobaleFundamentalSolutions}~{\em e)}
	\begin{align}\label{eq:RhoScalarProductSymmetry}
		\langle f,R_\la g\rangle
		=
		\langle f,(R_\la^--R_\la^+) g\rangle
		=
		\langle (R_\la^+-R_\la^-)f, g\rangle
		=
		-\langle R_\la f,g\rangle
		\,,
	\end{align}
	and thus the assignment \eqref{eq:Rho} is well-defined. Sesquilinearity and non-degenerateness is clear, and \eqref{eq:RhoSymmetry} follows directly from \eqref{eq:RhoScalarProductSymmetry}.
\end{proof}

Next we describe the solutions of $D_\la$ in a little more detail. This is a direct corollary of our preceding constructions.

\begin{corollary}
	\begin{enumerate}
		\item Let $\Sigma$ be a Cauchy hyperplane such that $K\subset\Sigma^+$ or $K\subset\Sigma^-$, and $u$ (smooth, compactly supported) Cauchy data on $u$. Then there exists precisely one solution $f_\la\in\Sol_\la$ with Cauchy data $u$ on $\Sigma$.
		\item Let $f_\la\in\Sol_\la$ be a solution of $D_\la$ and $\eps>0$ sufficiently small. Then there exist $g^\pm\in\Ccin(M^\pm_{\tau,\eps})$ such that $f_\la=R_\la g^+=R_\la g^-$ and 
		\begin{align}
			f_\la(x)
			&=
			(N_{\tau,\la}^+R^+g^-)(x)
			=
			(N_{\tau,\la}^-R^-g^+)(x)
			\,,\qquad
			x\in M_\tau\,.
		\end{align}
		The functions
		\begin{align}
			f_{\la,n}(x)
			:=
			\sum_{k=1}^n((-\la R^\pm W)^kR^\pm g^\mp)(x)
			\,,\qquad 
			x\in M_\tau,
		\end{align}
		converge to the restriction of $f_\la$ to $M_\tau$ as $n\to\infty$, in the topology of $\Ltwon(M_\tau)$.
	\end{enumerate}
\end{corollary}
\begin{proof}
	{\em a)} We may find a time slice neighborhood $\bSigma$ of $\Sigma$ such that $K\subset J^+(\bSigma)\backslash\bSigma$ or $K\subset J^-(\bSigma)\backslash\bSigma$. Then the unique solution $f_0[u]$ of $D$ with Cauchy data $u$ on $\Sigma$ can be written as $f_0[u]=Rg$, where $g\in\Ccin(\bSigma)$. Let $f_\la:=R_\la g$. Then $f_\la\in\Sol_\la$, and as $f_\la|_{\bSigma}=f_0[u]|_{\bSigma}$, $f_\la$ has Cauchy data $u$ on $\Sigma$. Uniqueness of this solution follows as in the proof of Proposition~\ref{proposition:SolutionSpace}~{\em a)}.
	
	{\em b)} This is immediate from our construction of the fundamental solutions $R_\la^\pm$.
\end{proof}
	
The results presented so far show a strong similarity to the well-known results in the solution theory of normally hyperbolic differential operators. We next show that despite this similarity, the Cauchy problem is in general ill-posed in the present context. This will be done with two examples.
	
\begin{example}\label{Example:CauchyNoSolution}{\bf (Cauchy Problem with no solution)}\\
	Let $D=\square$ be the d'Alembert operator, $\Sigma$ a Cauchy hyperplane, and $Wh:=\langle w_1,h\rangle w_2$ with $w_1,w_2\neq0$ such that $\supp w_1\subset\Oold_1$, $\supp w_2\subset \Oold_2$ with two spacelike separated double cones $\Oold_1$, $\Oold_2$ over~$\Sigma$ (see figure below). Pick Cauchy data $u$ on $\Sigma$ supported in $\Oold_1$ such that $f_0[u]$, the unique solution of $D$ with these Cauchy data, satisfies $\langle w_1,f_0[u]\rangle\neq0$, and also assume that $Rw_2\neq0$. Then there exists no $f_\la\in\Sol_\la$ with Cauchy data~$u$.
\end{example}
\begin{figure}[h]\label{figure:NoCauchy}
	\begin{center}
		\includegraphics[width=80mm]{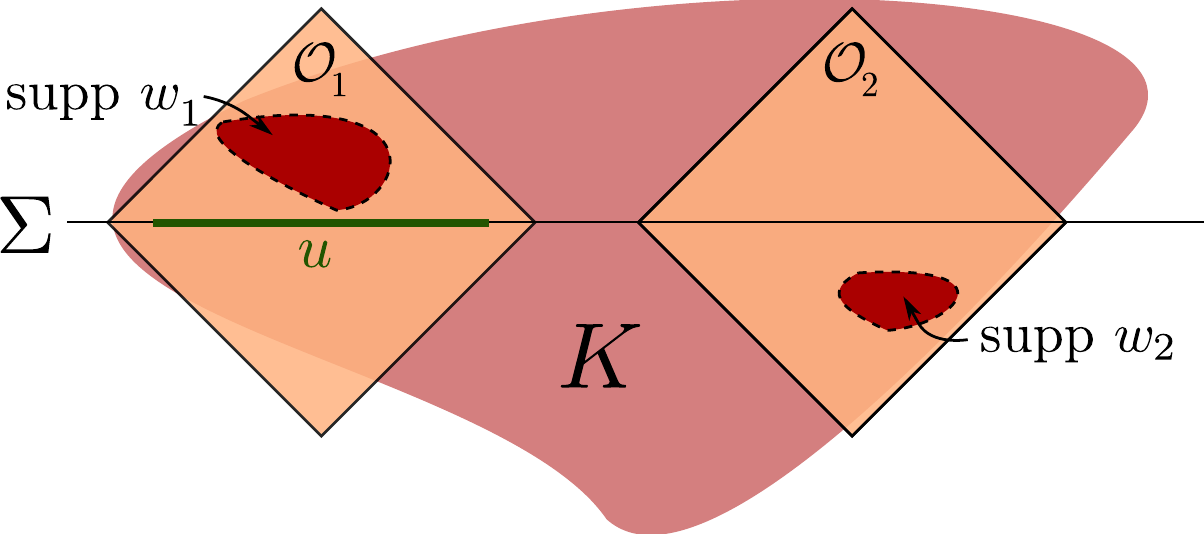}
		\\
		{\em Sketch of the geometric situation of Example~\ref{Example:CauchyNoSolution}.}
	\end{center}
\end{figure}
\begin{proof}
	We first observe that the assumptions made can easily be satisfied by suitably adjusting $w_1,w_2$. For a proof by contradiction, assume $f_\la\in\Sol_\la$ has Cauchy data $u=(u_0,u_1)$ supported only in $\Oold_1\cap\Sigma$. Due to the form of $Wh=\langle w_1,h\rangle \cdot w_2$ and the fact that $\supp w_2$ is disjoint from $\Oold_1$, one observes that $(Df_\la)(x)=0$, $x\in\Oold_1$. Hence the restriction of $f_\la$ to $\Oold_1$ is a solution of $D$, and as $\Oold_1$ is causally convex, this solution is uniquely determined by its Cauchy data $u$ (which are entirely contained in $\Oold_1$), i.e. $f_\la|_{\Oold_1}=f_0[u]|_{\Oold_1}$, where $f_0[u]$ is the solution of $D$ on $\Rl^n$ with Cauchy data $u$.
	
	We next determine $f_\la$ on $\Oold_2$. To this end, we use the equation \eqref{eq:GaussMitRandtermen} with the Cauchy hyperplane $\Sigma$ considered here. Inserting the special form of $W$ and $U^0=0$, and taking into account the supports of $w_1$, $w_2$, we obtain with $h\in\Ccin$ and $T_\la^\pm=R_\la^\mp$ since $D=D^*$,
	\begin{align*}
		\langle h,f_\la\rangle
		&-
		\la\left\{ 
			\overline{\langle w_2, R_\la^+h\rangle}\int_{\Sigma_t^+}dx\,( w_1(x),f_\la(x))
			+
			0
			-
			0
			-
			\langle w_1,f_\la\rangle\int_{\Sigma_t^-}dx\,((R_\la^-h)(x),w_2(x))
		\right\}
		\\
		&=
		\langle h,f_\la\rangle
		-
		\la\left\{ 
			\langle R_\la^+h,w_2\rangle \langle w_1,f_\la\rangle
			-
			\langle w_1,f_\la\rangle\langle R_\la^-h,w_2\rangle
		\right\}
		\\
		&=
		\langle h,f_\la\rangle
		+
		\la\langle h,R_\la w_2\rangle\langle w_1,f_\la\rangle
		\\
		&=
		\int_{\Sigma_t}
		\left\{((R_\la h)_t,u_1)-((\partial_0 R_\la h)_t,u_0)\right\}
		\,.
	\end{align*}
	For $h\in\Cin(\Oold_2)$, we have $W R^\pm h=0$ because $\supp w_1$ is spacelike to $\Oold_2$, and thus $R_\la h=Rh$. (This holds in particular for $h=w_2$.) Hence, by the assumption on the Cauchy data of $f_\la$, the integral over $\Sigma=\Sigma_t$ above vanishes for all such $h$. In view of the above equality, we then have $0=\langle h,f_\la+\la\langle w_1,f_\la\rangle R_\la w_2\rangle$, and since $h\in\Cin(\Oold_2)$ was arbitrary,
	\begin{align*}
		f_\la(x)
		=
		-\la\langle w_1,f_\la\rangle (R_\la w_2)(x)
		=
		-\la\langle w_1,f_0[u]\rangle (R w_2)(x)
		,\qquad
		x\in\Oold_2
		\,.
	\end{align*}
	Thus the Cauchy data of $f_\la$ and $Rw_2$ on $\Oold_2\cap\Sigma$ differ only by the (non-zero) factor $-\la\langle w_1,f_0[u]\rangle$. But by assumption, these Cauchy data are zero. Furthermore, the Cauchy data of $Rw_2$ on $\Sigma$ can have only support in $\Oold_2$ since $\supp w_2\subset\Oold_2$, and thus we conclude that $Rw_2$, as a solution of $D$, must vanish identically. This is a contradiction.
\end{proof}

\begin{example}{\bf (Cauchy Problem with non-unique solution)}\\
	Let again $W f=\langle w_1,f\rangle w_2$, with $w_1,w_2\in\Ccin$ with spacelike separated supports, and a Cauchy hyperplane $\Sigma$ such that $\supp w_2\subset \Sigma^-$. Denoting the Cauchy data of $Rw_2$ on $\Sigma$ by $u$, let $f_\la:=f_0[u]-R^+w_2$. Then $w_1$ and $\la\neq0$ can be chosen in such a way that the $R_\la^\pm$ exist, $f_\la$ is a non-zero solution of $D_\la$, and $f_\la$ has zero Cauchy data on $\Sigma$.
\end{example}
\begin{proof}
	We first note that by construction, both $f_0[u]$ and $R^+w_2$ have Cauchy data $u$ on $\Sigma$, and thus $f_\la$ has zero Cauchy data. As the support of $f_0[u]$ extends to infinitely late times, whereas $\supp(R^+w_2)\subset J^+(\supp w_2)$ extends only to the future, $f_\la$ is non-zero. We calculate
	\begin{align*}
		D_\la f_\la
		&=
		Df_0[u]-DR^+w_2+\la\,\langle w_1,f_0[u]-R^+ w_2\rangle w_2
		\\
		&=
		\left(-1+\la\langle w_1,f_0[u]-R^+ w_2\rangle\right)\cdot w_2
		\\
		&=
		\left(-1+\la\langle w_1,f_0[u]\rangle\right)\cdot w_2
		\,,
	\end{align*}
	and have to make sure that $\langle w_1,f_0[u]\rangle\neq0$, so that $D_\la f_\la=0$. This can be done by adjusting $w_1$ suitably. Then $f_\la$ is a non-zero solution of $D_\la$ for $\la=1/\langle w_1,f_0[u]\rangle\neq0$.
	
	Moreover, the fundamental solutions $R_\la^\pm$ exist for this value of $\la$. In fact, since $\supp w_1$ lies spacelike to $\supp w_2$, we have $WR^\pm W=0$, so that the Neumann series \eqref{eq:NeumannSeries} terminate, and thus converge for all $\la\in\Cl$.
\end{proof}

As these examples demonstrate, the Cauchy problem for $D_\la$ is in general ill-posed for Cauchy hyperplanes $\Sigma$ such that $K\not\subset\Sigma^+$ and $K\not\subset\Sigma^-$ --- both existence and uniqueness of solutions can fail. What can however be analyzed is the scattering of free solutions of $D$ at the perturbation $W$, closely related to the relative Cauchy evolution \cite{BrunettiFredenhagenVerch:2001,FewsterVerch:2012,HollandsWald:2001}. This is the topic of the next section.

\subsection{Scattering}
	
We have seen before that any solution $f_\la\in\Sol_\la$ restricts to free solutions $f_0^\pm\in\Sol_0$ in the future $\Sigma_{\tau_+}^+$ and past $\Sigma_{\tau_-}^-$ of the perturbation $W$. In general, it will not be possible to prescribe solutions $f_0^+$ and $f_0^-$ of $D$ such that $f_\la(x)=f^+_0(x)$ for $x_0>\tau_+$ and $f_\la(x)=f^-_0(x)$ for $x_0<\tau_-$, as $f_0^+$ is uniquely determined by $f_0^-$ and vice versa. The relation between these  ``incoming free asymptotics'' to ``outgoing free asymptotics'' -- where ``free''  refers here to the unperturbed differential operator $D$ -- is nothing but the scattering at the non-local potential $W$, which we are going to analyze next. We first define two M\o ller type operators 
\begin{align}
	\Om_{\la,\pm}&:\Sol_\la\to\Sol_0\,,\\
	\Om_{\la,\pm}&:R_\la g\mapsto R g\,,\qquad g\in\Ccin(\Sigma_{\tau_\pm}^\pm)
	.
	\label{def:Ompm}
\end{align}

\begin{proposition}\label{proposition:MollerOperators}{\bf (M\o ller operators)}
	\begin{enumerate}
		\item The M\o ller operators $\Om_{\la,\pm}$ are well-defined linear bijections.
		\item Given a solution $f\in\Sol_\la$, we have
		\begin{align}\label{eq:OmAndAsymptotics}
			f|_{\Sigma_{\tau_\pm}^\pm}=(\Om_{\la,\pm} f)|_{\Sigma_{\tau_\pm}^\pm}.
		\end{align}
		\item The M\o ller operators $\Om_{\la,\pm}$ intertwine the sesquilinear forms $\rho_0$ and $\rho_\la$ defined in \eqref{eq:Rho}, i.e.
		\begin{align}\label{eq:MollerIntertwiner}
			\rho_0(\Om_{\la,\pm} f_\la,\Om_{\la,\pm} g_\la)
			=
			\rho_\la(f_\la,g_\la)
			\,,\qquad
			f_\la,g_\la\in\Sol_\la\,.
		\end{align}
	\end{enumerate}
\end{proposition}
\begin{proof}
	{\em a)} Let $g\in\Ccin(\Sigma_{\tau_\pm}^\pm)$ and assume that $R_\la g=0$. Then in particular the restriction $(R_\la g)|_{\Sigma_{\tau_\pm}^\pm}$ vanishes. But on $\Sigma_{\tau_\pm}^\pm$, we have $0=(R_\la g)|_{\Sigma_{\tau_\pm}^\pm}=(R g)|_{\Sigma_{\tau_\pm}^\pm}$, and as $Rg$ is a solution of $D$, this implies that $Rg=0$ on all of $\Rl^n$. Hence the assignment \eqref{def:Ompm} is well-defined and injective.
	
	By Proposition~\ref{proposition:SolutionSpace}~{\em a)}, $\Sol_\la=R_\la\Ccin(\Sigma_{\tau_\pm}^\pm)$. Thus \eqref{def:Ompm} defines in fact a linear mapping from $\Sol_\la$ to $\Sol_0$, and also surjectivity is immediate from \eqref{def:Ompm} and Proposition~\ref{proposition:SolutionSpace}~{\em a)}.
	
	{\em b)} Let $f=R_\la g\in\Sol_\la$, $g\in\Ccin(\Sigma_{\tau_\pm}^\pm)$. Then $f|_{\Sigma_{\tau_\pm}^\pm}=Rg|_{\Sigma_{\tau_\pm}^\pm}$. This is the same as \eqref{eq:OmAndAsymptotics}.
	
	{\em c)} Let $f_\la,g_\la\in\Sol_\la$. By Proposition~\ref{proposition:SolutionSpace}~{\em a)}, we find $f^\pm,g^\pm\in\Ccin(\Sigma_{\tau_\pm}^\pm)$ such that $f_\la=R_\la f^+=R_\la f^-$, $g_\la=R_\la g^+=R_\la g^-$ and therefore, $\Om_{\la,\pm}f_\la=R f^\pm$, $\Om_{\la,\pm}g_\la=R g^\pm$. Thus, the left and right hand sides of \eqref{eq:MollerIntertwiner} can be written as
	\begin{align*}
		\rho_0(\Om_{\la,\pm}f_\la,\Om_{\la,\pm}g_\la)
		&=
		\rho_0(Rf^\pm,Rg^\pm)
		=
		\langle f^\pm,R g^\pm\rangle
		\,,\\
		\rho_\la(f_\la,g_\la)
		&=
		\rho_\la(R_\la f^\pm,R_\la g^\pm)
		=
		\langle f^\pm,R_\la g^\pm\rangle
		\,.
	\end{align*}
	But as in part {\em b)}, we have $(R_\la^\pm g^\pm)|_{\Sigma_{\tau_\pm}^\pm}=(R^\pm g^\pm)|_{\Sigma_{\tau_\pm}^\pm}$, and consequently $\langle f^\pm,R_\la g^\pm\rangle=\langle f^\pm,R g^\pm\rangle$.
\end{proof}

We can now introduce the {\em scattering operator}
\begin{align}\label{def:ScatteringOperator}
	S_\la
	:=
	\Om_{\la,+}(\Om_{\la,-})^{-1}:\Sol_0\to\Sol_0\,,
\end{align}
which maps the incoming asymptotics $\Om_{\la,-}f$ of a solution $f\in\Sol_\la$ to its outgoing asymptotics $\Om_{\la,+}f$, and thus describes the scattering by the potential term $\la W$ \cite{BrunettiFredenhagenVerch:2001}. 

\begin{theorem}\label{theorem:ScatteringOperator}{\bf (Scattering Operator)}
	\begin{enumerate}
		\item The scattering operator $S_\la:\Sol_0\to\Sol_0$ \eqref{def:ScatteringOperator} is a linear bijection.
		\item $S_\la$ preserves the sesquilinear form $\rho_0$ \eqref{eq:Rho}, i.e.
		\begin{align}
			\rho_0(S_\la f,S_\la g)=\rho_0(f,g)
			\,,\qquad
			f,g\in\Sol_0\,.
		\end{align}
		\item Explicitly, $S_\la$ is given by
		\begin{align}\label{eq:SFormel}
			S_\la
			=
			1+\la RWN_{\tau,\la}^+
			=
			1+RW\sum_{k=0}^\infty \la^{k+1}(- R^+W)^k
			\,.
		\end{align}
		The sum converges in the norm of bounded operators on $\Ltwon(M_\tau)$.
		\item For any $f_0\in\Sol_0$,
		\begin{align}
			\la\longmapsto S_\la f_0
		\end{align}
		is analytic in the topology of $\Cin$ on a finite disc around $\la=0$. In particular, $f_0\in\Sol_0$,
		\begin{align}
			\left.\frac{d (S_\la f_0)}{d\la}\right|_{\la=0}
			=
			RWf_0
			\,.
		\end{align}
	\end{enumerate}
\end{theorem}
\begin{proof}
	Part {\em a)} is clear from Proposition~\ref{proposition:MollerOperators}~{\em a)}. For {\em b)}, we use Proposition~\ref{proposition:MollerOperators}~{\em c)} and invertibility of the M\o ller operators to find, $f,g\in\Sol_0$,
	\begin{align*}
		\rho_0(S_\la f,S_\la g)
		=
		\rho_0(\Om_{\la,+}(\Om_{\la,-})^{-1}f,\Om_{\la,+}(\Om_{\la,-})^{-1}g)
		=
		\rho_\la((\Om_{\la,-})^{-1}f,(\Om_{\la,-})^{-1}g)
		=
		\rho_0(f,g)
		\,.
	\end{align*}
	{\em c)} Given a solution $f_0\in\Sol_0$,
	\begin{align*}
		\psi
		:=
		WN^+_{\tau,\la} f_0
		=
		W\sum_{k=0}^\infty(-\la R_\tau^+W)^kf_0
	\end{align*}
	is well-defined because each term is restricted to $K\subset M_\tau$ by the action of $W$. Since $W$ is smoothing, we have $\psi\in\Ccin(K)$. This implies that with $f_0$, also $f_0+\la R\psi$ is a solution of $D$, i.e $T_\la:=1+\la RWN_{\tau,\la}^+$ is a well-defined linear map $T_\la:\Sol_0\to\Sol_0$.
	
	Any solution of $D$ is uniquely determined by its restriction to $M_\tau$. To prove $T_\la=S_\la$, it is therefore sufficient to prove that the restrictions of $S_\la f_0$ and $T_\la f_0$ to $M_\tau$ coincide.

	To do so, we consider a solution $f_\la\in\Sol_\la$. Then we find $\eps>0$ and $g^\pm\in\Ccin(M_{\tau,\eps}^\pm)$ such that $f_\la=R_\la g^+=R_\la g^-$ and $f_0:=Rg^-=\Om_{\la,-}f_\la$ as well as $S_\la f_0=\Om_{\la,+}f_\la=Rg^+$, by definition of $\Om_{\la,\pm}$ and $S_\la$. As every solution $f_0\in\Sol_0$ arises in this way, what is left to prove is
	\begin{align}\label{eq:SvsT}
		(R g^+)|_{M_\tau}
		=
		(T_\la Rg^-)|_{M_\tau}
	\end{align}
	To this end, we compute 
	\begin{align*}
		f_\la|_{M_\tau}
		&=
		R_{\tau,\la} g^+
		\\
		&=
		R^-_{\tau,\la} g^+-R^+_{\tau,\la} g^+
		\\
		&=
		N_{\tau,\la}^-R^-_{\tau} g^+-R^+_{\tau} g^+
		\\
		&=
		(1-\la N_{\tau,\la}^-R^-_{\tau}W)
		R^-_{\tau} g^+-R^+_{\tau} g^+
		\\
		&=
		R_{\tau}g^+ - \la	N_{\tau,\la}^-R_{\tau}^-W R_{\tau}^-g^+
		\\
		&=
		R_{\tau}g^+ - \la	N_{\tau,\la}^-R_{\tau}^-W R_{\tau} g^+
		\\
		&=
		(1- \la	N_{\tau,\la}^-R_{\tau}^-W) R_{\tau} g^+
		\\
		&=
		N_{\tau,\la}^- R_{\tau} g^+
		\,,
	\end{align*}
	where we have used the definition of $R_{\tau,\la}^\pm$, equation \eqref{eq:N}, and the fact that because of the supports of $g^+$ and $W$, we have $WR_\tau^- g^+=WR_\tau g^+$. In complete analogy, one computes $f_\la|_{M_\tau} =N_{\tau,\la}^+R_\tau g^-$. 
	
	We thus have $(Rg^+)|_{M_\tau}=R_\tau g^+=(N_{\tau,\la}^-)^{-1}(f_\la|_{M_\tau})=(N_{\tau,\la}^-)^{-1}N_{\tau,\la}^+R_\tau g^-$. Using the equation $(1+X)(1+Y)^{-1}=1-(Y-X)(1+Y)^{-1}$, valid for operators $X,Y$ with $\|X\|,\|Y\|<1$, we find
	\begin{align*}
		(Rg^+)|_{M_\tau}
		&=
		(N_{\tau,\la}^-)^{-1}N_{\tau,\la}^+ R_\tau g^-
		\\
		&=
		(1+\la R_\tau^-W)(1+\la R_\tau^+W)^{-1} 
		R_\tau g^-
		\\
		&=
		(1+\la R_\tau W N_{\tau,\la}^+)R_\tau g^-
		\\
		&=
		(T_\la Rg^-)|_{M_\tau}
		\,.
	\end{align*}
	This shows \eqref{eq:SvsT} and thus $S_\la=T_\la$. The second equality in \eqref{eq:SFormel} follows by inserting the definition of $N_{\tau,\la}^+$.
	
	{\em d)} For $f_0\in\Sol_0$, we have $Wf_0\in\Ccin(K)\subset\Ltwon(M_\tau)$. As the Neumann series $N_{\tau,\la}^+$ converges in the norm of $\B(\Ltwon(M_\tau))$, the function
	\begin{align*}
		\la\longmapsto N_{\tau,\la}^+f_0
		=
		f_0-\la N_{\tau,\la}^+R_\tau^+ Wf_0
		\in\Ltwon(M_\tau)
	\end{align*}
	is analytic (in the norm topology of $\Ltwon(M_\tau)$) for sufficiently small $|\la|$. But $W:\Ltwon(M_\tau)\to\Ccin$ and $R:\Ccin\to\Cin$ are linear and continuous. Hence $\la\mapsto f_0+\la RWN_{\tau,\la}^+ f_0\in\Cin$ is analytic in the topology of $\Cin$. According to {\em c)}, this function coincides with $\la\mapsto S_\la f_0$.
	
	In view of this analyticity, we can differentiate under the sum in \eqref{eq:SFormel} and immediately obtain $\partial_\la S_\la f_0|_{\la=0}=RWf_0$, $f_0\in\Sol_0$.
\end{proof}

\section{Perturbations by star products}\label{Section:StarProducts}

In this section we discuss two examples of perturbations $W$ which are not $\Ccin$-kernel operators, but rather limits thereof. These examples arise in the context of (classical) field theory of noncommutative spaces, where one seeks to describe the dynamics in the presence of a noncommutatively coupled potential. For the case of a Dirac operator and a star product which is commutative in time, such an analysis was carried out in \cite{BorrisVerch:2010}. Here we can generalize to the case of noncommutative time.

We will present two examples, each of which violates one the important properties of $\Ccin$-kernel operators, namely either the smoothness or the compact support of the kernel. We will not fully analyze these perturbations here, but rather show how they fit in the framework described previously as limits of $\Ccin$-kernel operators, and that the important scattering derivation $\partial_\la S_\la|{\la=0}$ still exists here.

The basic structure we will be concerned with is that of Rieffel's product \cite{Rieffel:1992}. Thus the main ingredient is an action $\alpha$ of $\Rl^n$. In our context, $\alpha$ will act on various function spaces (for simplicity, we here take $N=1$, i.e. consider scalar functions) by pullback of an $\Rl^n$-action $\tau$ on $\Rl^n$, i.e. by $(\alpha_z f)(x)=f(\tau_z(x))$. Picking also an antisymmetric, invertible, real $(n\times n)$-matrix as deformation parameter, we consider products of the form
\begin{align}
	w\star f
	:=
	\int_{\Rl^n} dp\int_{\Rl^n} dz\,e^{2\pi i(p,z)}\,\alpha_{\te p}w\cdot\alpha_z f
	\,
	.
\end{align}
For our purposes, we will always take $w\in\Ccin$, and $f$ will be a smooth function on $\Rl^n$ with falloff properties depending on the choice of $\tau$.

The best known example is to take $\tau_z(x)=x+z$ and $f\in\Ss(\Rl^n)$ (Schwartz space). In this case $\star$ coincides with the Moyal product, and we have a continuous associative but noncommutative product $\star$.

Another class of examples has been discussed in \cite{LechnerWaldmann:2011}, see also \cite{BahnsWaldmann:2007,HellerNeumaierWaldmann:2006} for earlier related work. There the idea is to take $\tau$ of such a form that it leaves a compact set $K\subset\Rl^n$ invariant. In more detail, such an action $\tau$ can for example be constructed as follows \cite{HellerNeumaierWaldmann:2006,LechnerWaldmann:2011}. Let $\gamma:(-1,1)\to\Rl$ be a diffeomorphism, and define, $x=(x_1,...,x_n),z=(z_1,...,z_n)\in\Rl^n$,
\begin{align*}
	\tau_z(x)_k
	&:=
	\begin{cases}
		\gamma^{-1}(\gamma(x_k)+z_k) & |x_k|<1\\
		x_k & |x_k|\geq1
	\end{cases}
	\qquad.
\end{align*}
Clearly $\tau$ is an $\Rl^n$-action, and $K:=[-1,1]^n$ is invariant under $\tau$.
When $\gamma$ is appropriately chosen, $\tau$ is also smooth and polynomially bounded. We recall from \cite[Sect.~5]{LechnerWaldmann:2011} that this can be achieved by choosing $\gamma$ such that $\gamma$ is antisymmetric, $\gamma(x_k)=\exp(\frac{1}{1-x_k})$ for $x_k>\frac{1}{2}$, and $\gamma'(x_k)\geq\gamma'(0)>0$. In this case, one can take $f\in\Cin$, and again obtain a continuous associative but noncommutative product $\star$.
\\
\\
In both these situations, the one of the canonical translations $\tau$ and the the one just discussed, the integral has to understood as an oscillatory integral taking values in $\Ss(\Rl^n)$ and $\Cin(\Rl^n)$, respectively. Concretely, it can always be calculated according to 
\begin{align}\label{eq:w*f_eps}
	(w\star f)(x)
	=
	\lim_{\eps\to0}
	\int_{\Rl^n} dp\int_{\Rl^n} dz\,e^{2\pi i(p,z)}\,\chi(\eps p)\chi(\eps z)w(\tau_{\te p}(x)) f(\tau_z(x))
	\,,
\end{align}
where $\chi\in\Ccin(\Rl^n)$ is a cutoff function, equal to $1$ on an open neighborhood of $0$, and $w\star f$ is independent of the choice of $\chi$.

In the following, we will consider for $w\in\Ccin$ the perturbation term
\begin{align}\label{eq:W}
	Wf
	:=
	w\star f
	\,,
\end{align}
and denote by $W_\eps$ the integral operator in \eqref{eq:w*f_eps}, $\eps>0$. When the two cases need to be distinguished, we will also write $W_{(\eps)}^M$ and $W^K_{(\eps)}$, respectively. To simplify matters, we will also require that the support of $w$ is contained in the interior of $K$ in the case of $W^K$.

\begin{lemma}{\bf (Properties of star product kernels)}\\
	Let $w\in\Ccin$ and $W=W^M$ or $W=W^K$ defined as in \eqref{eq:W}.
	\begin{enumerate}
		\item Let $\eps>0$. Then $W_\eps$ is a $\Ccin$-kernel operator.
		\item The integral kernel of $W^M$ is smooth, but not of compact support.
		\item The integral kernel of $W^K$ is of compact support (in $K$), but not smooth.
	\end{enumerate}
\end{lemma}
\begin{proof}
	We first consider the Moyal product. Then we have
	\begin{align*}
		(W_\eps^M f)(x)
		&=
		\int_{\Rl^n} dp\int_{\Rl^n} dz\,e^{2\pi i(p,z)}\,\chi(\eps p)\chi(\eps z)w(x+\te p) f(x+z)
		\\
		&=
		\int_{\Rl^n} dy
		\left(
			\int_{\Rl^n} dp\,e^{2\pi i(p,(y-x))}\,\chi(\eps p)\chi(\eps (y-x))w(x+\te p)\right) f(y)
		\,.
	\end{align*}
	Because both $w$ and $\chi$ have compact support, $\chi(\eps p)w(x+\te p)$ vanishes for all $p$ if $|x|$ is large enough. Furthermore, $\chi(\eps(y-x))=0$ if $|x-y|$ is large enough. Thus $W^M_\eps$ has a kernel of compact support. The smoothness of this kernel follows from well-known statements on the Fourier transform, i.e. we have shown {\em a)} for the Moyal product. 
	
	For fixed $x$, the limit $\eps\to0$ can be taken under the integral, and after a change of variables one finds
	\begin{align*}
		(W^Mf)(x)
		=
		\frac{1}{(2\pi)^{n/2}|\det\te|}\int dy\,e^{ix\cdot\te^{-1}y}\widetilde{w}(\te^{-1}(y-x))\,f(y)
		\,,
	\end{align*}
	where $\widetilde{w}$ denotes the Fourier transform of $w$. From this formula it is obvious that the kernel of $W^M$ is smooth and not compactly supported, i.e. we have shown {\em b)}.
	
	Now we consider the locally noncommutative product, and observe the following two properties of $w\star f$: First, if $x\notin K$, then also $\tau_{\te p}(x)\notin K$ for all $p\in\Rl^n$ since $K$ is invariant and $\tau$ is an action. In view of $\supp w\subset K$, we then have $w(\tau_{\te p}(x))=0$ for all $p$, and hence $(w\star f)(x)=0$ for $x\notin K$. Second, if $\supp f\cap K=\emptyset$, and $x\in K$, then $f(\tau_z(x))=0$ for all $z\in\Rl^n$, and thus again $(w\star f)(x)=0$. These remarks apply to both $W^K$ and $W^K_\eps$ and show that these operators have kernels supported in $K\times K$.
	
	Explicitly, we find after a change of variables, $x\notin K$,
	\begin{align*}
		&(W^K_\eps f)(x)
		=
		\int_{\Rl^n} dp\int_{\Rl^n} dz\,
		e^{2\pi i(p,z)}\,
		\chi(\eps p)\chi(\eps z)\,
		w(\tau_{\te p}(x)) f(\tau_z(x))
		\\
		&=
		\int_K \frac{dy\,\gamma'(y)}{|\det\te|}f(y)
		\int_{\Rl^n} dp\,
		e^{2\pi i(\gamma(x),\te^{-1}\gamma(y))}\,
		e^{2\pi i(\te p,(\gamma(y)-\gamma(x))}
		\chi(\eps p)\chi(\eps \gamma(y)-\eps\gamma(x))\,
		w(\tau_{\te p}(x))
		\,,
	\end{align*}
	where we used the shorthand notations $\gamma(x):=(\gamma(x_1),...,\gamma(x_n))$ and $\gamma'(x):=(\gamma'(x_1),...,\gamma'(x_n))$. If the distance of $x$ to the boundary of $K$ is smaller than some minimal distance $d$ (depending on $\eps$, $\te$, and the supports of $\chi$, $w$), then $\chi(\eps p)w(\tau_{\te p}(x))=0$ for all $p$ by the support properties of $\chi$ and $w$. Furthermore, if the distance $\eps|\gamma(y)-\gamma(x)|$ is large enough, then $\chi(\eps \gamma(y)-\eps\gamma(x))=0$. Since $\gamma$ diverges as $x$ or $y$ approach the boundary of $K$, this implies that the integral over $p$ vanishes for $(x,y)$ outside some compact set properly contained in the interior of $K$. Thus $W^K_\eps$ is a $\Ccin$-kernel operator.
	
	However, the kernel of $W^K$ is not smooth. In fact, by a calculation analogous to the one for the Moyal product, one can compute that $W^K$ has the integral kernel
	\begin{align*}
		(W^Kf)(x)
		&=
		\frac{1}{(2\pi)^{n/2}|\det\te|}
		\int dy\,k(x,y)f(y)\,,\\
		k(x,y)
		&:=
		\begin{cases}
			\gamma'(y)\,
			e^{i\gamma(x)\cdot \te^{-1}\gamma(y)}\,
			\widetilde{\varphi}(\te^{-1}(\gamma(y)-\gamma(x))\big) & x,y\in K\\
			0 & x\notin K \text{ or } y\notin K
		\end{cases}
		\,,
	\end{align*}
	where $\varphi:=w\circ\gamma^{-1}$ and the tilde denotes a Fourier transform. From this formula, one sees that $k$ is discontinuous at the boundary of $K$, for example by noting that $k(x,x)=\gamma'(x)$ diverges as $x$ approaches the boundary of $K$ from the inside.
\end{proof}

Both star product operators, $W^M$ and $W^K$, are thus limits of $\Ccin$-kernel operators which do not lie in this class themselves. In the case of the locally noncommutative star product, the smoothness of the kernel fails, but each $W^K_\eps$ has support in the same set $K$. Here it is conceivable that our methods can be generalized in such a way that also $W^K$ can be analyzed along the same lines as $\Ccin$-kernel operators, i.e. that its fundamental solutions, scattering operator, etc. can be constructed. 

The Moyal star product operator $W^M$ differs more drastically from the situation considered so far, as the supports of $W^M_\eps$ grow infinitely as $\eps\to0$ (albeit they have smooth kernels). Here one would need to pass to an asymptotic formulation of the scattering problem.

These matters will be discussed in more detail elsewhere. However, already at the present stage one can show that the derivative of the scattering operator at zero coupling, which is the essential quantity for the connection to Bogoliubov's formula, does exist also for the locally noncommutative multiplier.

\begin{proposition}
	Let $W=W^K$ be defined as in \eqref{eq:W}. Then $WR_\tau^\pm$ and $R_\tau^\pm W$ extend to bounded operators on $\Ltwon(M_\tau)$.
\end{proposition}
\begin{proof}
	The pointwise product $f\mapsto w\cdot f$ is a continuous linear map $\Cin(M_\tau)\to\Cin(M_\tau)$. As $K$ is contained in $M_\tau$, the action $\alpha_x f:=f\circ\tau^K_x$ is a smooth polynomially bounded $\Rl^n$-action on $\Cin(M_\tau)$ \cite[Prop.~5.5]{LechnerWaldmann:2011}, and hence, $f\mapsto w\star f$ is also a continuous linear map $\Cin(M_\tau)\to\Cin(K)$ \cite[Prop.~4.6]{LechnerWaldmann:2011}. As $R_\tau^\pm:\Ccin(M_\tau)\to\Cin(M_\tau)$ are continuous, we see that $WR^\pm_\tau:\Ccin(M_\tau)\to\Cin(K)$ is continuous. But $\Cin(K)$ embeds continuously in $\Ltwon(M_\tau)$, and $\Ccin(M_\tau)\subset\Ltwon(M_\tau)$ is dense. Hence $WR^\pm_\tau$ extends continuously to $\Ltwon(M_\tau)\to\Ltwon(M_\tau)$. Similarly, $R^\pm_\tau W:\Ccin(M_\tau)\to\Cin(M_\tau)$ is continuous, and by the support properties of $R^\pm f$, we can also extend $R^\pm_\tau$ to $\Ltwon(M_\tau)$ as in Lemma~\ref{lemma:RestrictedGreenOperators}.
\end{proof}

In the case of a Moyal multiplier $W^M$, the space of functions on the time slice $M_\tau$ is not invariant under $W^M$. Rather, one needs to take $\tau\to\infty$ as $\eps\to0$ to guarantee that $W_\eps^M$ is always supported in $M_{\tau(\eps)}$. Some aspects of this limit have been studied in \cite{Borris:2011}.
\section{CCR and CAR Quantization}\label{Section:Quantization}

So far we have concentrated on the integro-differential equation $D_\la f=0$ for classical fields~$f$. However, as explained in the Introduction, most of our motivation comes from analyzing this equation for {\em quantum} fields. Thanks to the linearity of $D_\la$, a quantization of the solution spaces $\Sol_\la$ is possible in quite a straightforward manner, as we shall outline in this section. We will distinguish two cases: The case of a symmetric differential operator, which naturally leads to CCR quantization, and the case of a Dirac operator, which naturally leads to CAR quantization. For (pre-)normally differential operators $D$ without the non-local perturbation $W$, such an analysis has been carried out in~\cite{BarGinoux:2011}, and for Dirac operators with a perturbation which is local in time, see \cite{BorrisVerch:2010}.

We begin with the CCR case and assume that $D$ is a symmetric differential operator with symmetric perturbation $W=W^*$ (and real coupling $\la$). We also introduce a conjugation $C$ on $\Cl^N$, i.e. an antiunitary involution $C:\Cl^N\to\Cl^N$. By pointwise action, $C$ operates also on all function spaces appearing here, and will be denoted by the same symbol $C$ everywhere. Note that on $\Ltwon$, the so defined conjugation is antiunitary. To single out a real space of solutions, we assume that 
\begin{align}
	DCf=CDf
	\,,\qquad
	WCf=CWf
	\,,\qquad
	f\in\Ccin.
\end{align}
An example for this situation is given by the Klein-Gordon operator $D=\Box+V(x)$ and an integral operator $W_w$ \eqref{eq:CinftyIntegralOperator}, where $V\in C^\infty(\Rl^n,\Cl^{N\times N})$, $w\in C_0^\infty(\Rl^n\times\Rl^n,\Cl^{N\times N})$ are potentials which take only real values in some basis of $\Cl^N$, and $C$ is complex conjugation in that basis.

For real $\la$, then also $D_\la=D+\la W$ and $C$ commute, and we define
\begin{align}
	\Sol_{\la,C}
	:=
	\{f\in\Sol_\la\,:\,Cf=f\}
	\,.
\end{align}

\begin{proposition}{\bf (Symplectic structure of solution spaces)}\\
\vspace*{-8mm}
	\begin{enumerate}
		\item Under the assumptions made, the sesquilinear form $\rho_\la$ \eqref{eq:Rho} restricts to a real-valued real bilinear non-degenerate symplectic form $\sigma_\la:\Sol_{\la,C}\times\Sol_{\la,C}\to\Rl$.
		\item The M\o ller operators $\Om_{\la,\pm}$ and the scattering operator $S_\la$ restrict to symplectomorphisms $\Om_{\la,\pm}:(\Sol_{\la,C},\sigma_\la)\to(\Sol_{0,C},\sigma_0)$ and $S_\la:(\Sol_{0,C},\sigma_0)\to(\Sol_{0,C},\sigma_0)$, respectively.
	\end{enumerate}
\end{proposition}
\begin{proof}
	{\em a)} It is clear that $\sigma_\la$ is real bilinear, and using a decomposition into real and imaginary parts w.r.t. $C$, one also sees that $\sigma_\la$ inherits non-degenerateness from $\rho_\la$.
	
	As the conjugation $C$ preserves supports, we have $\supp(CR_\la^\pm Cf)=\supp(R_\la^\pm Cf)\subset J^\pm(\supp f)\cup J^\pm(K)$, $f\in\Ccin$. Furthermore, $CR_\la^\pm CDf=CR_\la^\pm DCf=C^2f=f$ and similarly, $DCR^\pm_\la Cf=f$. In view of the uniqueness of the fundamental advanced and retarded solutions (Proposition~\ref{proposition:Uniqueness}), we conclude $CR_\la^\pm C=R_\la^\pm$. This implies that $\sigma_\la$ is real-valued: For any $f,g\in\Ccin$, $Cf=f$, $Cg=g$, we have
	\begin{align*}
		\sigma_\la(R_\la f,R_\la g)
		=
		\langle f,R_\la g\rangle
		=
		\langle Cf,R_\la Cg\rangle
		=
		\overline{\langle f,C R_\la Cg\rangle}
		=
		\overline{\langle f, R_\la g\rangle}
		=
		\overline{\sigma_\la(R_\la f, R_\la g)}
		\,.
	\end{align*}
	Antisymmetry of $\sigma_\la$ follows now from \eqref{eq:RhoSymmetry}. 
	
	{\em b)} The definition \eqref{def:Ompm} of $\Om_{\la,\pm}$ and the fact that $CR_\la C=R_\la$, $CRC=R$ implies that $\Om_{\la,\pm}$ maps $\Sol_{\la,C}$ onto $\Sol_{0,C}$. The fact that these operators are symplectic follows from Proposition~\ref{proposition:MollerOperators}~{\em c)}. The analogous statements for $S_\la$ are easily deduced from \eqref{def:ScatteringOperator} and Theorem~\ref{theorem:ScatteringOperator}~{\em b)}.
\end{proof}

We thus have a real linear space $\Sol_{\la,C}$ endowed with a real-valued real bilinear non-degenerate symplectic form $\sigma_\la$. These data can now be used to proceed to a quantum field $\phi_\la$ satisfying the differential equation $D_\la\phi_\la=0$ by considering the corresponding CCR algebra $\frA_\la:=\CCR(\Sol_{\la,C},\sigma_\la)$ over $(\Sol_{\la,C},\sigma_\la)$ in a canonical manner \cite{BratteliRobinson:1997}.

On the level of the $C^*$-algebras $\frA_\la$, $\frA_0$, we have Bogoliubov isomorphisms $\alpha_{\la,\pm}:\frA_\la\to\frA_0$, induced by the M\o ller operators, and a scattering automorphism $s_\la:\frA_0\to\frA_0$, induced by $S_\la$. 
\\
\\
The second case we want to consider is the more particular case of a Dirac operator $D$ as an example of a prenormally hyperbolic operator (with $D'=-D$). Thus we take $D=-i\gamma^\mu\partial_\mu+V(x)$, where the $\gamma^\mu$ satisfy the Clifford relations \cite{Coquereaux:1988}, in particular, $(\gamma^0)^*=\gamma^0=(\gamma^0)^{-1}$ and $(\gamma^k)^*=-\gamma^k=\gamma^0\gamma^k\gamma^0$, $k=1,...,s$. We restrict to dimension $n$ even or $n=3\mod 8$ or $n=9\mod 8$. Setting $N:=2^{n/2}$ for $n$ even and $N:=2^{(n-1)/2}$ for $n$ odd, one can then also find a charge conjugation for the Dirac matrices $\gamma^\mu$, that is an antiunitary involution $C:\Cl^N\to\Cl^N$ satisfying $C\gamma^\mu C=-\gamma^\mu$, $\mu=0,1,...,s$. As before, we will use the same symbol $C$ to denote its pointwise action on functions taking values in $\Cl^N$.

We then have $C(-i\gamma^\mu\partial_\mu)C=-i\gamma^\mu\partial_\mu$, and upon requiring $CV(x)C=V(x)$, $CWC=W$, also $CD_\la C=D_\la$ (for $\la$ real). As before, this implies $CR_\la^\pm C=R_\la^\pm$.

Furthermore, the potential $V$ and the perturbation $W$ are required to satisfy $\gamma^0 V\gamma^0=V^*$, $\gamma^0 W\gamma^0=W^*$. In that case, we have ${D_\la}^*=\gamma^0 D_\la \gamma^0$, and thus ${R_\la}^*=-\gamma^0 R_\la \gamma^0$. Now we define
\begin{align}
	\delta_\la :\Sol_\la\times\Sol_\la &\to\Cl
	\,,\\
	\delta_\la(R_\la f,R_\la g)
	&:=
	i\rho_\la(R_\la\gamma^0 f,R_\la g)
	=
	i\langle f,\gamma^0 R_\la g\rangle
	\,.
	\label{eq:ScalarProductDirac}
\end{align}
\begin{proposition}{\bf (Hilbert space structure of Dirac field solution spaces)}
	\begin{enumerate}
		\item Under the assumptions made, $(\Sol_\la,\delta_\la)$ is a pre Hilbert space, with Hilbert space completion denoted $\K_\la$. If $\Sigma_t$ is a Cauchy hyperplane such that $K\subset\Sigma_t^+$ or $K\subset\Sigma_t^-$, then 
		\begin{align}\label{eq:DeltaPositive}
			\delta_\la(R_\la f,R_\la g)
			=
			\int_{\Sigma_t} ((R_\la f)_t,(R_\la g)_t)
			\,.
		\end{align}
		\item The M\o ller operators $\Om_{\la,\pm}$ and the scattering operator $S_\la$ extend to unitaries $\Om_{\la,\pm}:\K_\la\to\K_0$ and $S_\la:\K_0\to\K_0$, respectively.
		\item The conjugation $C$ induces an antiunitary operator $C_\la$ on $\K_\la$ by $C_\la R_\la f:=R_\la Cf$. We have
		\begin{align}\label{eq:CCommutingWithMollerAndS}
			\Om_{\la,\pm}C_\la
			=
			C_0\Om_{\la,\pm}
			\,,\qquad
			S_\la C_0=C_0S_\la\,.
		\end{align}
	\end{enumerate}
\end{proposition}
\begin{proof}
	{\em a)} As ${R_\la}^*=-\gamma^0 R_\la \gamma^0$, the mapping $\delta_\la(R_\la f,R_\la g)=-i\langle f,\gamma^0 R_\la g\rangle=i\langle R_\la f,\gamma^0g\rangle$ \eqref{eq:ScalarProductDirac} is well-defined, and clearly sesquilinear. 
	
	Furthermore, once \eqref{eq:DeltaPositive} is established, it is clear from this form that $\delta_\la$ is positive semidefinite. Actually, it is then definite: For if $\delta_\la(R_\la f,R_\la f)=0$, then the solution $R_\la f\in\Sol_\la$ vanishes on each Cauchy hyperplane $\Sigma_t$ such that $K\subset\Sigma_t$. Thus $R_\la f$ is supported in a light cone, which implies $R_\la f=0$.

	So it remains to show \eqref{eq:DeltaPositive}. To this end, we proceed in an analogous fashion as in the proof of Proposition~\ref{proposition:Uniqueness}, see also \cite[Prop.~1.2, Prop.~2.2]{Dimock:1982},  \cite[Prop.~2.1]{BorrisVerch:2010}, \cite[Lemma~3.17]{BarGinoux:2011}, for similar arguments for Dirac operators without non-local perturbation.
	
	Let $\Sigma_t$ be a Cauchy hyperplane such that $K\subset\Sigma_t^+$ (the case $K\subset\Sigma_t^-$ is analogous), $f,g\in\Ccin$, and consider the vector fields $X_\pm^\mu(x):=-((R_\la^{\pm}f)(x),\gamma^0\gamma^\mu (R_\la g)(x))$. Using the relations of the Clifford algebra and $V(x)^*=\gamma^0V(x)\gamma^0$, we compute
	\begin{align*}
		\partial_\mu X^\mu_\pm(x)
		&=
		-i((R_\la^\pm f)(x),\gamma^0(-i\gamma^\mu\partial_\mu (R_\la g)(x))
		+i
		((-i\gamma^\mu\partial_\mu(R_\la^\pm f)(x),\gamma^0(R_\la g)(x))
		\\
		&=
		-i((R_\la^\pm f)(x),\gamma^0(DR_\la g)(x))
		+
		i((DR_\la^\pm f)(x),\gamma^0(R_\la g)(x))
		\\
		&=
		i\la((R_\la^\pm f)(x),\gamma^0(W R_\la g)(x))
		-i\la ((WR_\la^\pm f)(x),\gamma^0(R_\la g)(x))
		+
		i(f(x),\gamma^0(R_\la g)(x))
		.
	\end{align*}
	The support of $X_\pm^\mu$ has compact intersection with $\Sigma_t^\mp$. Making use of Gauss' theorem, $\pm\int_{\Sigma_t} (X^0_\pm)_t=\int_{{\Sigma_t}^\mp}\partial_\mu X^\mu_\pm$, and taking into account $(\gamma^0)^2=1$, we get
	\begin{align*}
		\mp \int_{\Sigma_t} ((R_\la^\pm f)_t,(R_\la g)_t)
		&=
		i\int_{\Sigma_t^\pm}\,(f,\gamma^0 R_\la g)
		+
		i\la\int_{\Sigma_t^\pm}\,(R_\la^\pm f,\gamma^0 WR_\la g)
		-
		i\la\int_{\Sigma_t^\pm}\,(WR_\la^\pm f,\gamma^0 R_\la g)
		\\
		&=
		i\int_{\Sigma_t^\pm}\,(f,\gamma^0 R_\la g)
		\,,
	\end{align*}
	where in the second step, we have used $K\subset\Sigma^+$ and $\gamma^0W\gamma^0=W^*$. 
	
	Adding the equations for both choices of ``$\pm$'' gives the claimed equation \eqref{eq:DeltaPositive}:
	\begin{align*}
		\int_{\Sigma_t} ((R_\la f)_t,(R_\la g)_t)
		=
		i\langle f,\gamma^0 R_\la g\rangle
		=
		\delta_\la(R_\la f,R_\la g)
		\,.
	\end{align*}
	{\em b)} Let $f^\pm\in\Ccin(\Sigma_{\tau^\pm}^\pm)$. Then $\Om_{\la,\pm}R_\la f^\pm=R f^\pm$, and with Proposition~\ref{proposition:MollerOperators}~{\em c)}, we get
	\begin{align*}
		\delta_0(\Om_{\la,\pm} R_\la f^\pm,\Om_{\la,\pm} R_\la g^\pm)
		&=
		\delta_0(R f^\pm, R g^\pm)
		=
		i\rho_0(R\gamma^0 f^\pm,Rg^\pm)
		\\
		&=
		i\rho_\la(R_\la\gamma^0 f^\pm,R_\la g^\pm)
		=
		\delta_\la(R_\la f^\pm,R_\la g^\pm)
		\,.
	\end{align*}
	As the spaces $R_\la \Ccin\Sigma_{\tau^\pm}^\pm)\subset\K_\la$ and $R_0\Ccin(\Sigma_{\tau^\pm}^\pm)\subset\K_0$ are dense by construction of $\K_\la$, $\K_0$, this shows that $\Om_{\la,\pm}$ are linear isometries with dense domains and ranges, and therefore extend to unitaries. The scattering operator is also unitary as the product of $\Om_{\la,+}$ and ${\Om_{\la,-}}^{-1}$.
	
	{\em c)} As $R_\la$ and $C$ commute, $C_\la$ is well defined, and in view of
	\begin{align*}
		\delta_\la(C_\la R_\la f,C_\la R_\la g)
		=
		i\langle Cf,\gamma^0 R_\la Cg\rangle
		=
		\overline{-i\langle f,C\gamma^0 CR_\la g\rangle}
		=
		\overline{i\langle f,\gamma^0 R_\la g\rangle}
		=
		\overline{\delta_\la(R_\la f,R_\la g)}
		,
	\end{align*}
	also antiunitary. The commutation relations \eqref{eq:CCommutingWithMollerAndS} follow directly from the definition of the M\o ller operators and $CR_\la C=R_\la$.
\end{proof}

In the case of a perturbed Dirac operator, we have thus constructed a family of Hilbert spaces $\K_\la$ with antiunitary involutions $C_\la$, and can use these data to proceed to the CAR algebras $\frF_\la:=\CAR(\K_\la,C_\la)$ over $(\K_\la,C_\la)$ \cite{BratteliRobinson:1997}, similarly to the CCR case for symmetric $D_\la$. Again we have induced Bogoliubov isomorphisms $\alpha_{\la,\pm}:\frF_\la\to\frF_0$ on the level of the $C^*$-algebras $\frF_\la$ and $\frF_0$, induced by the M\o ller operators, and a scattering automorphism $s_\la:\frF_0\to\frF_0$, induced by $S_\la$. These structures form the prerequisites for a systematic study of the corresponding quantum Dirac fields.

In quantum field theory, the essential physical information is contained in the local structure of the algebra $\frA_\la$ (or $\frF_\la)$, i.e. a net $\Oold\mapsto\frA_\la(\Oold)$ mapping sub spacetimes $\Oold$ of $\Rl^n$ to sub algebras $\frA_\la(\Oold)$ of $\frA_\la$ \cite{Haag:1996}. Note that although the algebras $\frA_\la$ and $\frA_0$ (or $\frF_\la,\frF_0)$ are isomorphic ``globally'', i.e. as $C^*$-algebras, they are not ``locally isomorphic'' in the sense that the isomorphisms $\alpha_{\la,\pm}$ do not carry the subalgebras $\frA_\la(\Oold)$, $\frA_0(\Oold)$ into each other. In fact, one expects that the effects of the noncommutative perturbation $W$ manifest themselves precisely on the level of these nets by deviations from the usual local and covariant QFT setting. These matters will be studied in more detail in a forthcoming publication.

\small
\newcommand{\etalchar}[1]{$^{#1}$}


\end{document}